\documentclass[11pt]{article}

\usepackage{amsmath}
\usepackage{amsthm}
\usepackage{amssymb}
\usepackage{times}
\usepackage{helvet}
\usepackage{bbm}
\usepackage{courier}
\usepackage{graphicx}
\usepackage{wrapfig}
\usepackage{nicefrac}
\usepackage{hyperref}
\usepackage[vlined,ruled,linesnumbered]{algorithm2e}
\usepackage[nameinlink,capitalise]{cleveref}
\usepackage{paralist}
\hypersetup{%
    bookmarks=false,    
    pdftitle={Multiwinner Elections When Each Voter Ranks Only Few},    
    pdfauthor={Matthias Bentert and Piotr Skowron},                     
    colorlinks=true,       
    linkcolor=blue,       
    citecolor=blue,       
    filecolor=black,        
    urlcolor=blue,        
    linktoc=page            
}

\allowdisplaybreaks

\usepackage{todonotes}

\newcommand{\score}{{{\mathrm{score}}}}
\newcommand{\totscore}{{{\mathrm{total\_score}}}}
\newcommand{\ranked}{{{\mathrm{ranked}}}}

\newcommand{\reals}{{{\mathbb{R}}}}

\newtheorem{theorem}{Theorem}

\newtheorem{corollary}[theorem]{Corollary}

\newtheorem{definition}{Definition}
\newtheorem{example}{Example}

\newcommand{\calF}{\mathcal{F}}
\newcommand{\calA}{\mathcal{A}}
\newcommand{\calD}{\mathcal{D}}
\newcommand{\calR}{\mathcal{R}}

\DeclareMathOperator*{\argmin}{arg\,min}
\DeclareMathOperator*{\best}{best}
\DeclareMathOperator*{\worst}{worst}
\DeclareMathOperator*{\avg}{avg}
\DeclareMathOperator*{\m}{mod}

\crefname{eq}{equality}{equalities}

\newcommand{\minimax}{{{{\mathrm{MM}}}}}

\newcommand{\trunc}{{{{\mathrm{trunc}}}}}
\newcommand{\prob}{{{{\mathrm{P}}}}}
\newcommand{\ind}{{{{\mathrm{ind}}}}}
\newcommand{\expected}{{{{\mathrm{E}}}}}
\newcommand{\pos}{{{{\mathrm{pos}}}}}

\renewcommand{\sc}{{{{\mathrm{sc}}}}}

\DeclareMathOperator{\occ}{occ}

\title{Comparing Election Methods Where Each Voter Ranks Only Few Candidates}

\author{Matthias Bentert\\
  TU Berlin\\
  Berlin, Germany
  \and 
Piotr Skowron\\
  University of Warsaw\\
  Warsaw, Poland
}

\date{}

\begin{document}

\maketitle

\begin{abstract}
Election rules are formal processes that aggregate voters preferences, typically to select a single candidate, called the winner. Most of the election rules studied in the literature require the voters to rank the candidates from the most to the least preferred one. This method of eliciting preferences is impractical when the number of candidates to be ranked is large. We ask how well certain election rules (focusing on positional scoring rules and the Minimax rule) can be approximated from partial preferences collected through one of the following procedures: (i) randomized---we ask each voter to rank a random subset of $\ell$ candidates, and (ii) deterministic---we ask each voter to provide a ranking of her $\ell$ most preferred candidates (the $\ell$-truncated ballot). We establish theoretical bounds on the approximation ratios, and we complement our theoretical analysis with computer simulations. We find that mostly (apart from the cases when the preferences have no or very little structure) it is better to use the randomized approach. While we obtain fairly good approximation guarantees for the Borda rule already for~$\ell = 2$, for approximating the Minimax rule one needs to ask each voter to compare a larger set of candidates in order to obtain good guarantees.    
\end{abstract}

\section{Introduction}

An election rule is a function that takes as input a collection of voters preferences over a given set of $m$ candidates and returns a single candidate, called the winner. There is a large variety of election rules known in the literature (we refer the reader to the survey by Zwicker~\cite{zwi:b:intro-voting} for an overview); most of them require the voters to provide strict linear orders over the candidates. Yet, it is often hard, or even infeasible for a voter to provide such a prefernce ranking, especially when the set of candidates is large. Indeed, it is often believed that a voter can rank at most five to nine candidates~\cite{Mil56}.

In this paper we ask how the quality of decisions made through voting depends on the amount of information available. Specifically, our goal is to assess the quality of outcomes of elections when each voter can be asked to rank at most $\ell < m$ candidates. We compare two ways of eliciting preferences. In the first approach---which we call \emph{randomized}---we ask each voter to rank a random subset of $\ell$ candidates. In the second approach---which we call \emph{deterministic}---we ask each voter to provide the ranking of her top $\ell$ most preferred candidates (the, so-called, $\ell$-truncated ballot). For a number of rules (we analyze positional scoring rules and the Minimax method), we investigate how well they can be approximated by algorithms that use one of the two elicitation methods.

\subsection*{Our Contribution}

Our contribution is the following:
\begin{enumerate}
\item In \Cref{sec:scoring_rules} we identify a class $\mathit{Sep}_{\ell}$ of positional scoring rules that, for a given $\ell$, can be well approximated using the randomized approach. $\mathit{Sep}_{2}$ consists of a single rule, namely the Borda count; the number of rules in $\mathit{Sep}_{\ell}$ grows exponentially with $\ell$. We theoretically prove approximation guarantees for the rules from $\mathit{Sep}_{\ell}$---these guarantees are more likely to be accurate when the number of voters is large---we analytically show how, in the worst case, the approximation guarantees depend on the number of voters. In \Cref{sec:randomized:minimax} we provide an analogous analytical analysis for the Minimax rule.

\item In \Cref{sec:deterministic} we prove upper-bounds on the approximation ratios of an algorithm that uses $\ell$-truncated ballots; we prove these bounds both for positional scoring rules and for the Minimax rule. In both cases, we show that the algorithm that minimizes the maximal regret of Lu and Boutilier~\cite{lu-bou:maximizing-regret} (we recall this algorithm in \Cref{sec:_determnistic_algorithm_pos}) matches our upper-bounds (for Minimax our analysis is tight up to a small constant factor).

\item We ran computer simulations in order to verify how the approximation ratio depends on the particular distribution of voters preferences (\Cref{sec:experiments}). Our experiments confirm that in most cases (with the exception of very unstructured preferences) the randomized approach is superior. We also show that usually only a couple of hundreds of voters are required to achieve a reasonably good approximation.    
\end{enumerate}

\subsection*{Related Work}

Our work contributes to the broad literature on handling incomplete information in voting---for a survey on this topic, we refer the reader to the book chapter by Boutilier and~Rosenschein~\cite{bou-ros:b:incomplete_voting}. Specifically, our research is closely related to the idea of minimizing the maximal regret~\cite{lu-bou:maximizing-regret}. Therein, for a partial preference profile $P$, the goal is to select a candidate $c$ such that the score of $c$ in the worst possible completion of $P$ is maximized. In particular, algorithms minimizing the maximal regret yield the best possible approximation ratio. Our paper complements this literature by (1) providing an accurate analysis of these approximation ratios for various methods (which allows to better judge suitability of different methods for handling incomplete information), and (2) by providing the analysis for two natural methods of preference elicitation (which also allows to assess which of the two methods is better).

Algorithms for minimizing the maximal regret interpret the missing information in the most pessimistic way: they assume the worst-possible completion of partial preferences. Other approaches include assuming the missing pairwise preferences to be distributed uniformly (e.g.\ Xia and Conitzer~\cite{xia-con:c:mle-partial-orders}) and machine-learning techniques (Doucette~\cite{Dou14,Dou15}) to ``reconstruct'' missing information (assuming that the missing pairwise comparisons are distributed similarly as in observed partial rankings).

Our work is also closely related to the literature on distortion~\cite{pro-ros:c:distortion, car-pro:j:embeddings, bou-etal}. There, an underlying utility model is assumed, and the goal is to estimate how well various voting rules that have only access to ordinal preferences, approximate optimal winners, i.e., candidates that maximize the total utility of the voters. The concept of distortion has recently received a lot of attention in the literature. The definition of distortion has for example been adapted to social welfare functions (where the goal is to output a ranking of candidates rather than a single winner)~\cite{ben-pro-qui:distortion_welfare_functions} and to participator budgeting~\cite{ben-nat-pro-sha:participatory_budgeting_elicitation}. Some works also study distortion assuming a certain structure of the underlying utility model (e.g., that it can be represented as a metric space)~\cite{ans-bha-elk-pos-sko:j:distortion, AP16, FFG, GKM16, GAX17}.   

Finally, we mention that our randomized algorithms are similar to the one proposed by Hansen~\cite{Han16}. The main difference is that the rule proposed by Hansen asks each voter to compare a certain number of pairs of candidates, while in our approach we ask each voter to rank a certain fixed-size subset of them. Hansen views his algorithm as a fully-fledged standalone rule (and compares it with other election systems, mostly focusing on assessing the probability of selecting the Condorcet winner), while our primary goal is to investigate how well our rules approximate their original counterparts. 

\section{Preliminaries}
An election is a pair $E = (V, C)$, where $V = \{v_1, v_2, \ldots, v_n\}$ and $C = \{c_1, c_2, \ldots, c_m\}$ denote the sets of $n$ \emph{voters} and $m$ \emph{candidates}, respectively. Each voter $v_i$ is endowed with a \emph{preference ranking} over the candidates, which is a total ordering of the candidates and which we denote by~$\succ_i$. For each candidate $c\in C$ by $\pos_i(c)$ we denote the position of $c$ in $v_i$'s preference ranking. The position of the most preferred candidate is one, of the second most preferred candidate is two,~etc. For example, for a voter $v_i$ with the preference $c_2 \succ_i c_3 \succ_i c_1$, we have $\pos_i(c_1) = 3$, $\pos_i(c_2) = 1$, and $\pos_i(c_3) = 2$.

For an integer $t$ we use $[t]$ to denote the set $\{1, 2, \ldots, t\}$ and we use the Iverson bracket notation---for a logical expression~$P$ the term $[P]$ means $1$ if $P$ is true and $0$ otherwise.

A \emph{voting rule} is a function that, for a given election $E$, returns a subset of candidates, which we call tied \emph{winning candidates}. Below we describe several (classes of) voting rules that we will focus on in this paper.

A \emph{positional scoring function} is a mapping $\lambda\colon [m] \to \reals$ that assigns to each position a real value: intuitively, $\lambda(p)$ is a score that a voter assigns to a candidate that she ranks as her $p$-th most preferred one. For each positional scoring function $\lambda$ we define the $\lambda$-score of a candidate $c$ as $\sc_{\lambda}(c) = \sum_{v_i \in V} \lambda(\pos_i(c))$, and the corresponding election rule selects the candidate(s) with the highest $\lambda$-score. Examples of common positional scoring rules include:
\begin{description}
\item[Borda rule:] Based on a linear decreasing positional scoring function, the Borda rule is formally defined by $\beta(p) = m-p$ for $p \in [m]$.
\item[Plurality rule:] Being equivalent to the $1$-approval rule, the positional scoring function for the Plurality rule assigns a score of one to the first position and zero to all others.
\end{description}
 
Another important class of voting rules origins from the Condorcet criterion. It says that if there exists a candidate $c$ that is preferred to any other candidate by a majority of voters, then the voting rule should select $c$. We focus on one particular rule satisfying the Condorcet criterion (we chose a rule picking the candidates that maximize a certain scoring function so that we could apply to the rule the standard definition of approximation):
\begin{description}
\item[Minimax rule.] For an election $E = (V, C)$ and two candidates $c, c' \in C$, we define $\sc_{\minimax}(c, c') = |\{v_i \in V \mid c \succ_i c'\}|$ as the number of voters who prefer $c$ to $c'$ and we set
\begin{align*}
\sc_{\minimax}(c) = \min_{c' \neq c} \{\sc_{\minimax}(c, c')\} \text{.}
\end{align*}
The rule then selects the candidates with the highest $\sc_{\minimax}$ score. 
\end{description}

Since all rules described above select the candidates with the maximal scores (with particular rules differing in how the score should be calculated), a natural definition of approximation applies.
\begin{definition} 
We say that $\calA$ is an $\alpha$-approximation algorithm for a rule $\calR$ if for each election instance $E$ it holds that:
\begin{align*}
\frac{\score_{\calR}(\calA(E))}{\max_{w \in \calR(E)}\score_{\calR}(w)} \geq \alpha \text{,}
\end{align*}
where $\score_{\calR}$ is a function representing the score $\calR$ awards each candidate, $\calR(E)$ is the set of winners returned by $\calR$, and $\calA(E)$ is the candidate returned $\calA$. 
\end{definition}  

Later on, we will consider algorithms that have access only to certain parts of the input instances. In such cases the above definition still applies. For example, let $\trunc(E, \ell)$ denote the $\ell$ truncated instance obtained from $E$, i.e., a partial election which for each voter contains her preferences ranking from $E$, truncated to the top $\ell$ positions. Then we say that $\calA$ is an $\alpha$-approximation algorithm for $\calR$ for $\ell$-truncated instances, when for each election instance $E$ it holds that:
\begin{align*}
\frac{\score_{\calR}(\calA(\trunc(E)))}{\max_{w \in \calR(E)}\score_{\calR}(w)} \geq \alpha \text{.}
\end{align*}

\section{Randomized Approach}\label{sec:randomized}

In this section we explore a randomized approach, where each voter can be asked to rank a random subset of candidates.  

\subsection{Scoring Rules}\label{sec:scoring_rules}

We start our analysis by looking at the class of positional scoring rules. For the sake of simplicity we will assume throughout this section that $n$ is divisible by $m$\footnote{We will always implicitly assume that $n$ is much larger than $m$, and we will use randomized algorithms only. Thus, if $m$ does not divide $n$, then in our algorithms we can add a preliminary step that randomly selects a set of $n' = \left\lfloor \frac{n}{m} \right\rfloor \cdot m$ voters, and ignores the remaining $n - n'$ ones. We mention that other authors also suggested to give multiple randomized ballots to each voter.}. We first present an algorithm that estimates the score of each candidate and picks the candidate with the highest score. The algorithm is parameterized with a natural number $\ell \leq m$ and a vector of $\ell$ reals $\alpha = (\alpha_1, \ldots, \alpha_{\ell})$---for a fixed vector $\alpha$ we will call the algorithm $\alpha$-\textsc{PSF-ALG}. This algorithm asks each voter to rank a random set of $\ell$ candidates. We say that a candidate $c$ is ranked by a voter $v$ if $c$ belongs to the set of $\ell$ candidates that $v$ was asked to rank. If $c$ is the $i$-th most preferred among the candidates ranked by a voter, then $c$ receives the score of $\alpha_i$ from the voter. Such scores are summed up for each candidate, normalized by the number of voters who ranked the respective candidate, and the candidate with the highest total score is declared the winner. Pseudcode of the algorithm is given in \Cref{alg:psf}.

\begin{algorithm}[t]
\ForEach{candidate~$c$}
{
$\totscore(c) \leftarrow 0$ 

$\ranked(c) \leftarrow 0$
}
\ForEach{voter~$v$}
{
$S_v \leftarrow$ random set of $\ell$ candidates 

ask $v$ to rank $S_v$
 
\ForEach{$c \in S_v$}
{

\If{$c$ is ranked $i$-th among $S_v$}{ 

$\totscore(c) \leftarrow \totscore(c) + \alpha_i$ 

$\ranked(c) \leftarrow \ranked(c) + 1$
}
}
}
\ForEach{candidate~$c$}
{
$\score(c) \leftarrow \frac{n \cdot \totscore(c)}{\ranked(c)}$ 

}
\Return{candidate $c$ with maximal $\score(c)$}
\caption{Algorithm $\alpha$-\textsc{PSF-ALG} for computing positional scoring functions.}
\label{alg:psf}
\end{algorithm}

Below, we will show that for some positional scoring rules, by choosing the vector~$\alpha$ carefully, we can find good approximations of winning candidates with high probability. First, through \Cref{thm:separable} we establish a relation between positional scoring functions $\lambda$ and vectors $\alpha$ that should be used to assess $\lambda$; the formula is not intuitive, and we will discuss it later on. In particular, we will explain which positional scoring functions can be well approximated using this approach, that is, we will discuss the structure of the class of positional scoring functions which are covered by the following theorem.

\begin{theorem}\label{thm:separable}
Fix a non-increasing sequence of~$\ell$ reals~$\alpha = (\alpha_1, \ldots, \alpha_{\ell})$ and consider the positional scoring function~$\lambda_{\alpha}$ defined by
\begin{align*}
\lambda_{\alpha}(p) = \frac{1}{{m-1 \choose \ell-1}} \cdot \sum_{i = 1}^{\ell}\alpha_i {p-1 \choose i-1} \cdot {m - p \choose \ell - i} \text{.}
\end{align*}
For a candidate $c \in C$ that is ranked by at least one voter, we denote by $X_{c}$ the random variable describing the total normalized score that $c$ was assigned by $\alpha$-\textsc{PSF-ALG}. 
Then, the expected value $\expected(X_c)$ is equal to the $\lambda_{\alpha}$-score of $c$, and the probability that the score computed by $\alpha$-\textsc{PSF-ALG} for $c$ differs from its expected value by a multiplicative factor of $1 \pm \epsilon$ is upper-bounded by~$2\exp\left(-\frac{\epsilon^2 \expected(X_c)}{3}\right)$, i.e.,
\begin{align*}
p_{\epsilon} = \prob\Big(\left| X_c - \expected(X_c) \right| \geq \epsilon \expected(X_c) \Big) \leq 2 \exp\left(- \frac{\epsilon^2 \ell \sc_{\lambda_{\alpha}}(c)}{6m\alpha_1}\right) \text{.}
\end{align*}
\end{theorem}
\begin{proof}
Let us fix a candidate $c \in C$ who is ranked by at least one voter. The process of computing the score of $c$ according to \cref{alg:psf} can be equivalently described as follows. We first decide on the number~$x$ of voters we ask to rank $c$. Second, we pick uniformly at random a set $V'$ of $x$ voters such that all voters in~$V'$ are asked to rank~$c$ and all voters in~$V \setminus V'$ are not asked to rank~$c$. Finally, we ask each voter from $V'$ to rank $c$ and a randomly selected set of $\ell-1$ candidates.
Let $N_c$ be a random variable describing the number of voters who rank~$c$.
Further, for each voter $v$, let $X_{v, c, i}$ denote the random variable equal 1 if $c$ is the $i$-th candidate among those ranked by voter $v$ and zero otherwise. In particular, $X_{v, c, i}$ is zero when $v$ is not asked to rank $c$. 
Observe that if~$N_c = x$, then the value of~$X_c$ can be expressed as
\begin{align*}
X_{c} = \frac{n}{x} \cdot \sum_{v \in V} \sum_{i = 1}^\ell \alpha_i X_{v, c, i}  \text{.}
\end{align*}
Further, let $A_{V', c}$ be 1 if \emph{each} voter from $V'$ ranks $c$ and 0 otherwise. Similarly, let $A_{v, c}$ be~1 if $v$ ranks $c$ and 0 otherwise.  Let $\ind_v(S, c, i)$ be equal to 1 if $c$ is ranked as the $i$-th most preferred candidate among $S$ by $v$ and 0 otherwise.
We next compute the conditional expected value $\expected(X_{c}| N_c = x)$.
We first give the formal equalities and give reasoning for the more complicated ones afterwards.
\begin{align}
\expected(X_{c}| N_c = x) &= \expected\left(\frac{n}{x} \cdot \sum_{v \in V} \sum_{i = 1}^\ell \alpha_i X_{v, c, i} | N_c = x \right)
                          = \frac{n}{x} \cdot \sum_{v \in V} \sum_{i = 1}^\ell \alpha_i \expected( X_{v, c, i} | N_c = x) \label{equ:lam0}\\
\begin{split}
                          &= \frac{n}{x} \cdot \sum_{v \in V} \sum_{i = 1}^\ell \alpha_i \sum_{\substack{V'\subseteq V \\ |V'| = x \\ v\in V'}}\prob(A_{V', c} = 1| N_c = x) \expected( X_{v, c, i} | A_{V', c} = 1) \\
                          &= \frac{n}{x} \cdot \sum_{i = 1}^\ell \alpha_i \sum_{v \in V} \sum_{\substack{V'\subseteq V \\ |V'| = x}} \frac{1}{{n \choose x}} [v \in V'] \expected( X_{v, c, i} | A_{V', c} = 1) \\  
                          &= \frac{n}{x} \cdot \sum_{i = 1}^\ell \alpha_i \sum_{v \in V} \sum_{\substack{V'\subseteq V \\ |V'| = x}} \frac{1}{{n \choose x}} [v \in V'] \prob(A_{v, c} = 1 | A_{V', c} = 1) \expected( X_{v, c, i} | A_{v, c} = 1) \\
\end{split}\label{equ:lam0b}\\
                          &= \frac{n}{x} \cdot \sum_{i = 1}^\ell \alpha_i \sum_{v \in V} \frac{1}{{n \choose x}} \cdot {n-1\choose x-1} \expected( X_{v, c, i} | A_{v, c} = 1) \label{equ:lam1} \\     
                          &= \sum_{i = 1}^\ell \alpha_i \sum_{v \in V} \expected( X_{v, c, i} | A_{v, c} = 1) \label{equ:lam1b}\\     
				&= \sum_{i = 1}^{\ell}\alpha_i \sum_{v \in V} \frac{1}{{m - 1 \choose \ell - 1}}  \sum_{\substack{S \subseteq C\\ |S| = \ell}} [c \in S] \cdot \ind_v(S, c, i) \label{equ:lam3}\\
				&= \frac{1}{{m - 1 \choose \ell - 1}} \sum_{v \in V} \sum_{i = 1}^{\ell}\alpha_i {\pos_v(c)-1 \choose i-1} \cdot {m - \pos_v(c) \choose \ell - i} \label{equ:lam4}\\
				&= \sum_{v \in V} \lambda_{\alpha}(\pos_v(c)) = \sc_{\lambda_{\alpha}}(c) \label{equ:lam5} \text{.}   
\end{align}

We will now explain some of the equalities in the above sequence. (\ref{equ:lam1}) is an effect of regrouping the summands; each summand $\expected( X_{v, c, i} | A_{v, c} = 1)$ in the previous line is added for each set $V' \subseteq V$ of size $x$ which includes $v$---there are ${n-1 \choose x-1}$ such sets and $[v \in V'] \cdot \prob(A_{v, c} = 1 | A_{V', c} = 1)$ is the same as~$[v \in V']$.
(\ref{equ:lam3}) holds for the following reason: A voter who ranked $c$ was asked to rank some set of~$\ell$ candidates including~$c$. Each possible set has the same probability of being selected, thus this probability is~$\nicefrac{1}{{m-1 \choose \ell-1}}$.
(\ref{equ:lam4}) is true as we will show that~$\sum_{S \subseteq C} [c \in S] [|S| = \ell] \ind_v(S,c,i) = {\pos_v(c)-1 \choose i-1} \cdot {m - \pos_v(c) \choose \ell - i}$. Consider a fixed voter $v$, a fixed candidate $c$, and a set $S\subseteq C$ such that 
\begin{inparaenum}[(i)]
\item $c \in S$,
\item $|S| = \ell$, and
\item $v$ considers $c$ to be her~$i$-th most preferred candidate from $S$.
\end{inparaenum}
Each such a set must consist of~$i-1$ candidates that are ranked before~$c$ by~$v$ and~$\ell-i$ candidates that are ranked after~$c$. Thus, there are~${\pos_v(c)-1 \choose i-1} \cdot {m - \pos_v(c) \choose \ell - i}$ such sets. We refer to \cref{fig:crank} for an illustration.

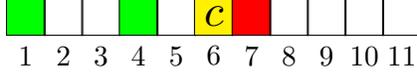
\begin{figure}
\centering
\begin{tikzpicture}
\node[draw,rectangle, minimum size=1, label=below:$1$, scale=2, fill=green] at (0,0) (n1) {};
\node[draw,rectangle, minimum size=0.5, label=below:$2$, scale=2] at (0.5,0) (n2) {};
\node[draw,rectangle, minimum size=0.5, label=below:$3$, scale=2] at (1,0) (n3) {};
\node[draw,rectangle, minimum size=0.5, label=below:$4$, scale=2, fill=green] at (1.5,0) (n4) {};
\node[draw,rectangle, minimum size=0.5, label=below:$5$, scale=2] at (2,0) (n5) {};
\node[draw,rectangle, minimum size=0.5, label=below:$6$, fill=yellow] at (2.5,0) (n6) {\scalebox{1.55}{$c$}};
\node[draw,rectangle, minimum size=0.5, label=below:$7$, scale=2, fill=red] at (3,0) (n7) {};
\node[draw,rectangle, minimum size=0.5, label=below:$8$, scale=2] at (3.5,0) (n1) {};
\node[draw,rectangle, minimum size=0.5, label=below:$9$, scale=2] at (4,0) (n1) {};
\node[draw,rectangle, minimum size=0.5, label=below:$10$, scale=2] at (4.5,0) (n1) {};
\node[draw,rectangle, minimum size=0.5, label=below:$11$, scale=2] at (5,0) (n1) {};
\end{tikzpicture}
\caption{An illustration explaining that there are~${\pos_v(c)-1 \choose i-1} \cdot {m - \pos_v(c) \choose \ell - i}$ possible sets~$S$ that satisfy~$c \in S$,~$|S| = \ell$ and~$c$ is ranked at position~$i$ in~$S$ by~$v$. Each field indicates a candidate and the number below a field indicates the position. The yellow field indicates~$c$, a red field indicates a candidate in~$S$ that was ranked after~$c$ by~$v$ and a green field indicates a candidate in~$S$ that was ranked before~$c$ by~$v$. A white field indicates a candidate that was not ranked by~$v$. In this example~${m=11},\ i = 3,\ \pos_v(c) = 6$, and~$\ell = 4$.}
\label{fig:crank}
\end{figure}

Next, we will use the Chernoff's inequality to assess the probability that the computed score of a candidate~$c$ does not differ from its true score by a factor of~$\epsilon$. We will first assess the conditional probability $\prob\Big(\left| X_c - \expected(X_c) \right| \geq \epsilon\expected(X_c) | N_c = x \Big)$. Observe that the conditional variables $\{X_{v, c, i} | N_c\}_{v \in V, i \in [\ell]}$ are not independent. For instance, if $X_{v, c, i} = 1$, then $X_{v, c, j} = 0$ for each $j \neq i$. However, they are all negatively correlated---intuitively meaning that if a variable becomes 1 (resp., 0), then the other variables are less likely to become 1 (resp., 0).
Thus, we can still apply the Chernoff's bound~\cite[Theorem 1.16, Corollary 1.10]{auger2011theory} which states that for any negatively-correlated random variables~$X_1,\ldots,X_n \in [0,1]$ such that~$X = \sum_{i=1}^n X_i$ and any~$\delta \in [0,1]$ it holds that
\begin{align}
\prob(X \leq (1-\delta) \expected(X)) \leq \exp(-\delta^2\expected(X)/2) \text{ and }
\prob(X \geq (1+\delta) \expected(X)) \leq \exp(-\delta^2\expected(X)/3). \label{ref:chernoff}
\end{align}
It follows immediately that~$\prob(|X - \expected(X)| \geq \delta \expected(X)) \leq 2\exp\left(-\delta^2\expected(X)/3\right)$.

Now, consider the variables $\frac{X_{v, c, i}\alpha_i}{\alpha_1}$. These variables are from $[0, 1]$ and from \cref{equ:lam0,equ:lam0b,equ:lam1,equ:lam1b,equ:lam3,equ:lam4,equ:lam5} we get that
\begin{align*}
\expected\left(\sum_{v \in V} \sum_{i = 1}^\ell \frac{X_{v, c, i}\alpha_i}{\alpha_1}| N_c = x \right) = \frac{x}{n\alpha_1}\sc_{\lambda_{\alpha}}(c) \text{.}
\end{align*}

This yields
\begin{align*}
&\prob\Big(\left| X_c - \expected(X_c) \right| \geq \epsilon\expected(X_c) | N_c = x \Big) \\
&\qquad = \prob\Big(\left|  \frac{x}{n\alpha_1} X_c - \expected\left( \frac{x}{n\alpha_1} X_c\right) \right| \geq \epsilon\expected\left( \frac{x}{n\alpha_1}X_c\right) | N_c = x \Big)
\overset{\cref{ref:chernoff}}\leq 2\exp\left(- \frac{\epsilon^2 x \sc_{\lambda_{\alpha}}(c)}{3n\alpha_1}\right).
\end{align*}

Finally, using the binomial identity~$(x+y)^n = \sum_{k=1}^{n} {n \choose k}x^k y^{n-k}$ we get
\begin{align*}
\prob\Big(\left| X_c - \expected(X_c) \right| \geq\epsilon\expected(X_c)\Big) &= \sum_{x = 0}^{n} \prob(N_c = x) \cdot \prob\Big(\left| X_c - \expected(X_c) \right| \geq \epsilon\expected(X_c) | N_c = x \Big) \\
          &\leq  \sum_{x = 0}^{n} {n \choose x} \cdot \frac{{{m-1 \choose \ell-1}}^x \left({m \choose \ell} - {m-1 \choose \ell -1}\right)^{n-x}}{{{m \choose \ell}}^n} \cdot 2\exp\left(- \frac{\epsilon^2 x \sc_{\lambda_{\alpha}}(c)}{3n\alpha_1}\right) \\
          &= 2 \exp\left(- \frac{\epsilon^2 \sc_{\lambda_{\alpha}}(c)}{3n\alpha_1}\right) \cdot \sum_{x = 0}^{n} {n \choose x} \cdot \left(\frac{\ell}{m}\right)^x \left(1 - \frac{\ell}{m}\right)^{n-x} \cdot e^{-x} \\
          &= 2 \exp\left(- \frac{\epsilon^2 \sc_{\lambda_{\alpha}}(c)}{3n\alpha_1}\right) \cdot \left(1 - \frac{\ell}{m} + \frac{\ell}{em} \right)^n \\
          &\leq 2 \exp\left(- \frac{\epsilon^2 \sc_{\lambda_{\alpha}}(c)}{3n\alpha_1}\right) \cdot \left(1 - \frac{\ell}{2m} \right)^n \\
          &\leq 2 \exp\left(- \frac{\epsilon^2 \sc_{\lambda_{\alpha}}(c)}{3n\alpha_1}\right) \cdot e^{-\frac{\ell n}{2m}} = 2 \exp\left(- \frac{\epsilon^2 \ell \sc_{\lambda_{\alpha}}(c)}{6m\alpha_1}\right) \text{.}
\end{align*}
This concludes the proof.
\end{proof}

Now, let us discuss the form of positional scoring functions $\lambda_{\alpha}(p)$ used in the statement of \Cref{thm:separable}. First, observe that for $\ell = 2$, if we set $\alpha_2 = 0$ and $\alpha_1 = 1$ we have that $\lambda_{\alpha}(p) = {p-1 \choose 0} \cdot {m - p \choose 2 - 1} = m - p = \beta(p)$. This means that by asking each voter to rank only two candidates, we can correctly (in expectation) assess the Borda scores of the candidates.
   
\begin{corollary}\label{cor:Borda_rand}
For a candidate $c$ the expected value of the score computed by Algorithm $(1, 0)$-\textsc{SEP-ALG} for $c$ is the Borda score of $S$.
\end{corollary}

Unfortunately, not every positional scoring function can be efficiently assessed while asking each voter to rank only few candidates. For example, we can generalize \Cref{cor:Borda_rand} and show that for any vector of two elements $\alpha = (\alpha_1, \alpha_2)$, the algorithm $\alpha$-\textsc{SEP-ALG} can only compute scores that are affine transformations of the Borda scores (thus, for $\ell=2$ the algorithm can only be used to approximate the Borda rule). 

We will now describe the class of all positional scoring functions which can be computed correctly in expectation by our algorithm for any fixed~$\ell$.
Since each positional scoring function is based on some~$m$-dimensional vector~$\beta = (\beta_1,\beta_2,\ldots, \beta_m)$ which can be expressed as~$\sum_{i=1}^{m} \alpha_i \cdot \beta_i$, where~$\alpha_1 = (1,0,\ldots)$, $\alpha_2 = (0,1,0,\ldots)$ and so on, these~$\alpha$-vectors form a basis of the linear space of positional scoring functions.

Let $\mathit{Sep}_{\ell} = \{\lambda_{\alpha} \colon \alpha \in \reals^\ell\}$ be the set of all positional scoring functions that can be computed (correctly in expectation) by our algorithm for a fixed $\ell$.
Since it holds for each two $\ell$-element vectors $\alpha, \alpha' \in \reals^{\ell}$ that~$\lambda_{\alpha + \alpha'}(p) = \sum_{i = 1}^{\ell}(\alpha_i + \alpha'_i) {p-1 \choose i-1} {m - p \choose \ell - i} = \sum_{i = 1}^{\ell}\alpha_i {p-1 \choose i-1} {m - p \choose \ell - i} + \sum_{i = 1}^{\ell}\alpha'_i {p-1 \choose i-1} \cdot {m - p \choose \ell - i} = \lambda_{\alpha}(p) + \lambda_{\alpha'}(p),$ we have that $\mathit{Sep}_{\ell}$ is a linear space too.
 
Thus, $\mathit{Sep}_{\ell}$ is an $\ell$-dimensional linear subspace of the $m$-dimensional space of all positional scoring functions, and so we can compactly describe it by providing $\ell$ scoring functions forming a basis of $\mathit{Sep}_{\ell}$.
\Cref{fig:depicted_basis} visually illustrates the scoring functions forming a basis for $\ell \in \{2, 4, 8\}$.
In other words, for a given value of $\ell$, we can use \Cref{thm:separable} to correctly compute (in expectation) all scoring functions which can be obtained as linear combinations of the scoring functions depicted in \Cref{fig:depicted_basis}.
   
\begin{figure}[!t]
\minipage{0.32\textwidth}
  \centering
  \includegraphics[width=\linewidth]{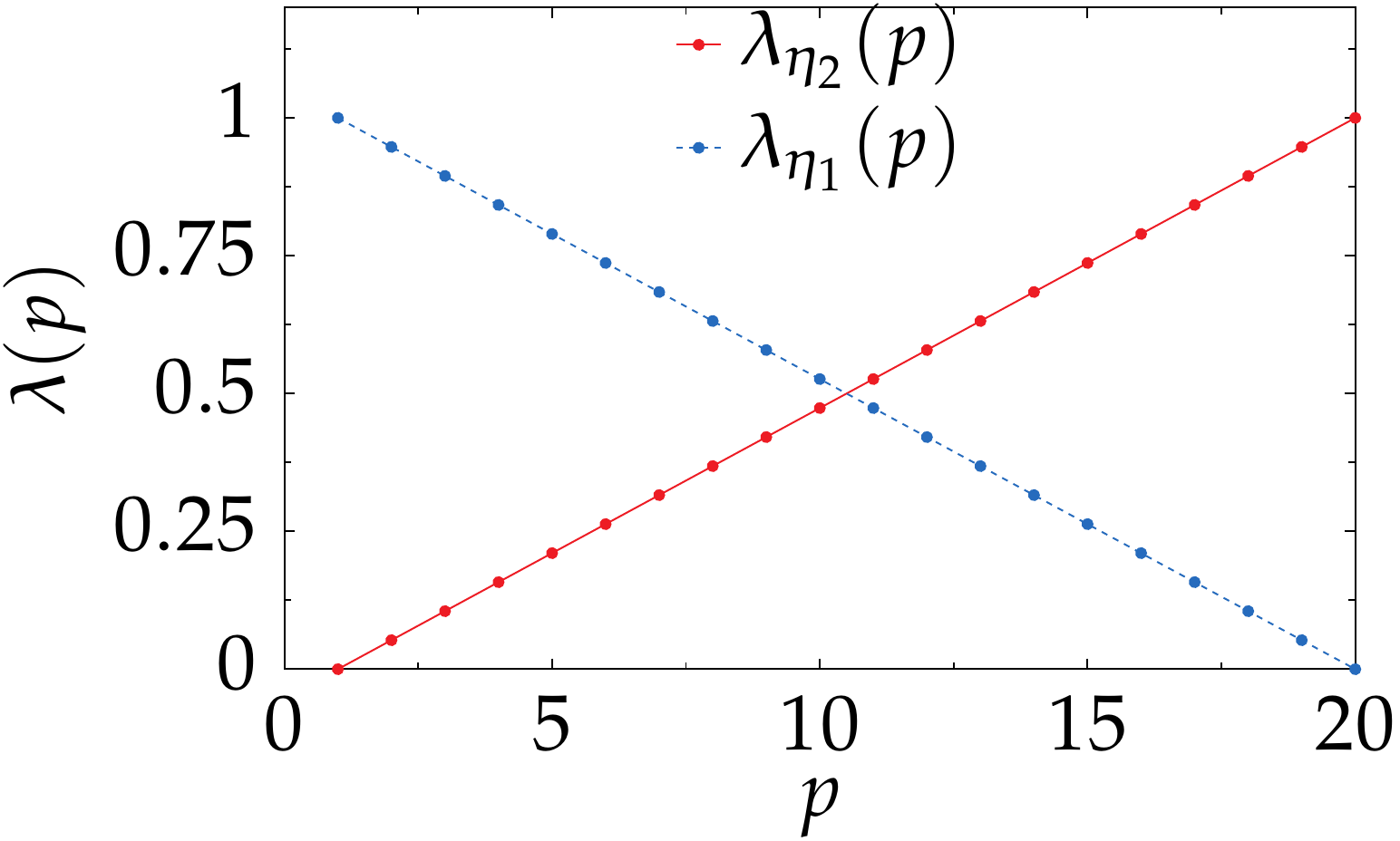}
  (a) $\ell = 2$
\endminipage\hfill
\minipage{0.32\textwidth}
  \centering
  \includegraphics[width=\linewidth]{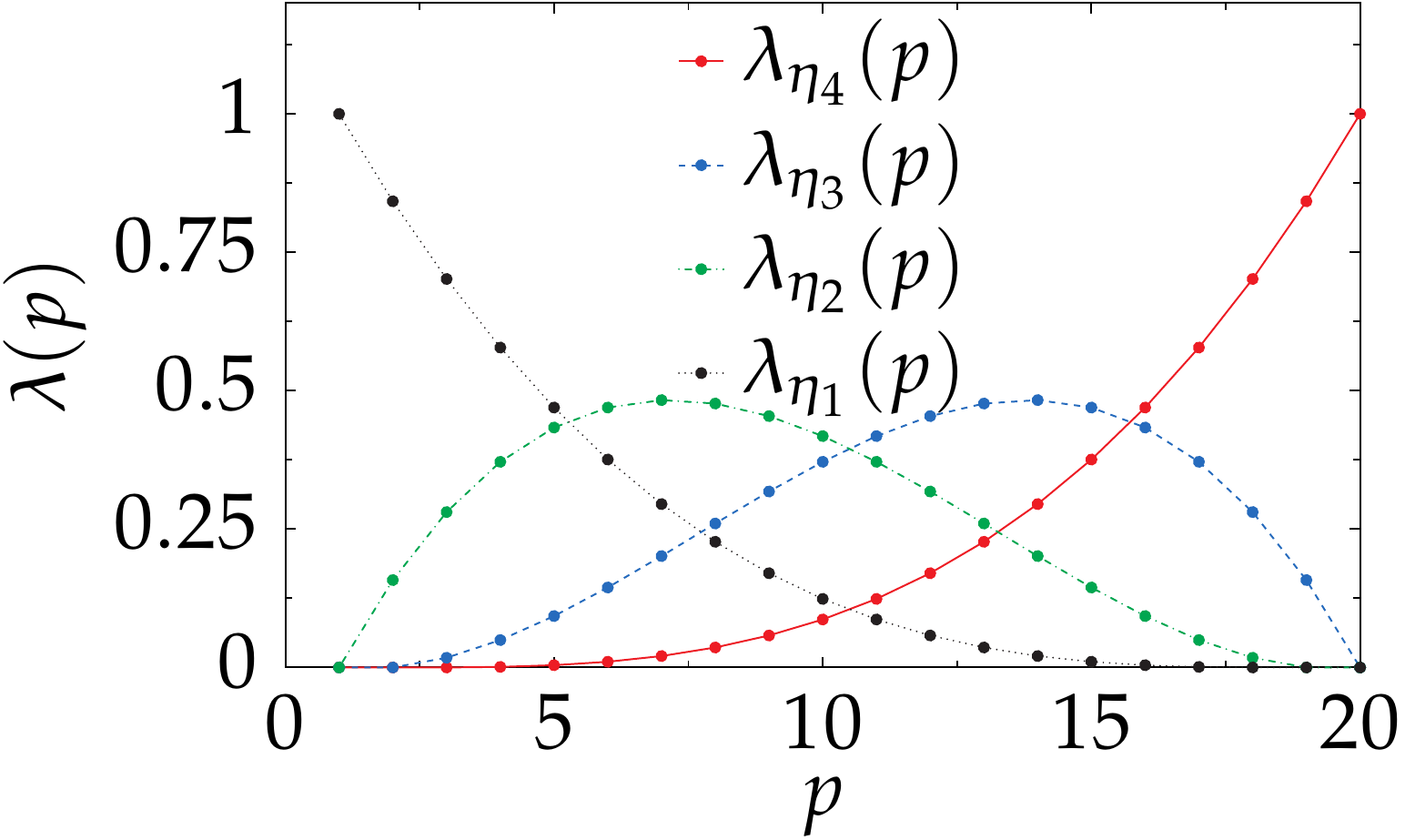}
  (b) $\ell = 4$
\endminipage\hfill
\minipage{0.32\textwidth}
  \centering
  \includegraphics[width=\linewidth]{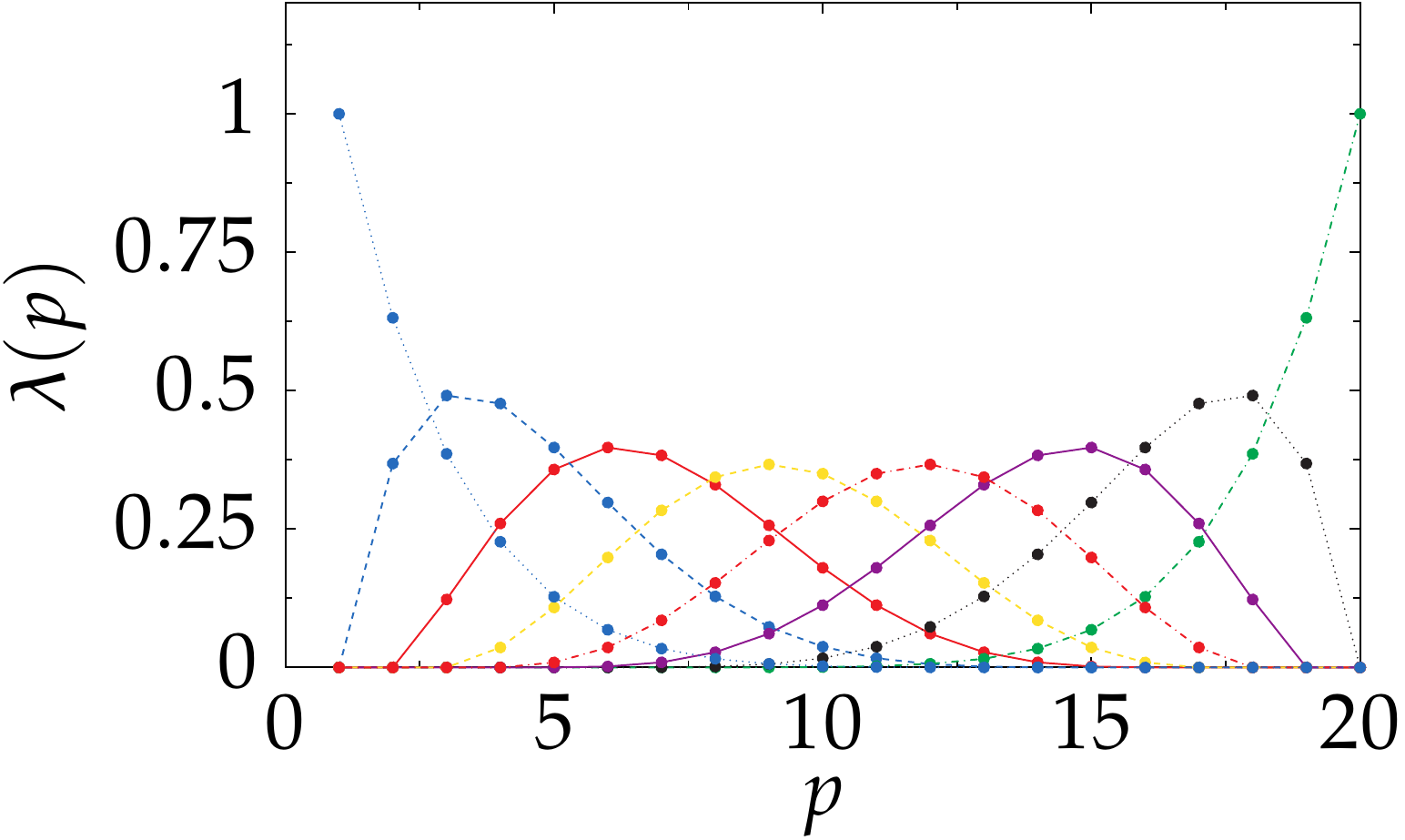}
  (c) $\ell = 8$
\endminipage
\caption{A basis of the space $\mathit{Sep}_{\ell} = \{\lambda_{\alpha} \colon \alpha \in \reals^\ell\}$ for three different values of $\ell$ (to compute the basis we took vectors $\alpha \in \{\eta_i \colon i \in [\ell]\}$, where $\eta_i$ is a vector with 1 in the $i$-th position, and 0 in the remaining $\ell-1$ positions). Each plot illustrates $\ell$ PSFs that span the space of all PSFs which can be correctly assessed by algorithm $\alpha$-\textsc{PSF-ALG} while asking each voter to rank only $\ell$ candidates.}\label{fig:depicted_basis}
\end{figure}

Finally, let us give some intuition regarding the probabilities assessed in \Cref{thm:separable}. For example, for $m = 21$ candidates and $n$ voters the Borda score of a winning candidate is at least $10n$. Assume that we want to ask each voter to compare only two candidates, and set $\epsilon = 0.01$. When assessing the score of a winning candidate, to get $p_{\epsilon} < 0.001$ we need about 72 thousands voters. For one million voters, this probability drops below $\nicefrac{4}{10^{42}}$.

Finally, note that \Cref{thm:separable} applies to any candidate, not only to election winners. This makes the result slightly more general, since it also applies to e.g., social welfare functions, where the goal is to output a ranking of the candidates instead of a single winner. 
 
\subsection{Minimax Rule}\label{sec:randomized:minimax}

We will now investigate whether the Minimax rule can be well approximated when each voter is only asked to rank a few candidates.
We will use an algorithm similar to \Cref{alg:psf}: each voter $v$ ranks a subset of candidates $S_v$ and whenever two candidates $c, c' \in S_v$ are ranked by a voter~$v$, we use her preference list to estimate $\sc_{\minimax}(c, c')$.
Notably, we scale the values $\sc_{\minimax}(c, c')$ for each two candidates $c, c' \in C$ by the number of times they were compared and use these normalized values to compute the Minimax winner.
This algorithm is formalized in \Cref{alg:minimax}.

\begin{algorithm}[t]
\ForEach{candidates~$c, c' \in C$}
{
$\mathrm{S}[c, c'] \leftarrow 0$ 
}
\ForEach{voter~$v$}
{
$S_v \leftarrow$ random set of $\ell$ candidates 

ask $v$ to rank $S_v$
 
\ForEach{$c \in S_v$}
{
\ForEach{$c' \in S_v \setminus \{c\}$}
{
\If{$c \succ_i c'$}{ 

$\mathrm{S}[c, c'] \leftarrow \mathrm{S}[c, c'] + 1$ 
}
}
}
}
\ForEach{candidate~$c$}
{
$\mathrm{S}[c] \leftarrow n\cdot \min \Big\{\frac{\mathrm{S}[c, c']}{\mathrm{S}[c, c'] + \mathrm{S}[c', c]}\colon c' \neq c, \mathrm{S}[c, c'] + \mathrm{S}[c', c] > 0 \Big\}$ 

}
\Return{candidate $c$ with maximal $\mathrm{S}[c]$}
\caption{Algorithm for computing winners according to the Minimax rule.}
\label{alg:minimax}
\end{algorithm}

\begin{theorem}
For each candidate $c \in C$ the probability that the total normalized score computed by \Cref{alg:minimax} for $c$ differs from the true Minimax score of $c$ by a multiplicative factor of at least $1 \pm \epsilon$ is upper-bounded by:
\begin{align*}
m\exp\left(- \frac{\epsilon^2 \ell^2 \sc_{\minimax}(c, c_{\min})}{6m^2}\right)
\end{align*}
\end{theorem}
\begin{proof}
First, let us fix a pair of candidates $c, c' \in C$ and let $X_{c, c'}$ be the random variable describing the value $n\cdot \frac{\mathrm{S}[c, c']}{\mathrm{S}[c, c'] + \mathrm{S}[c', c]}$ as computed by \Cref{alg:minimax}. Similarly as in the proof of \Cref{thm:separable} we can express $X_{c, c'}$ as a sum of negatively correlated random variables.

Specifically, computing $\mathrm{S}[c, c']$ according to \cref{alg:minimax} can be equivalently described as follows:
First, we decide on how many voters will be asked to compare $c$ and $c'$.
Let $N_{c,c'}$ be the random variable describing this number of voters.
Second, assuming $N_{c,c'} = x$, we pick uniformly at random a set $V'$ of $x$ voters, and we ask them to compare $c$ and $c'$.
For each voter $v$, let $X_{v, c, c'}$ denote the random variable equal 1 if voter $v$ said that she prefers $c$ to $c'$, and 0 otherwise. In particular, $X_{v, c, c'}$ is zero when $v$ is not asked to compare $c$ and $c'$. 
Observe that if~$N_{c,c'} = x$, then $\mathrm{S}[c, c'] + \mathrm{S}[c', c] = x$, and so the value of~$X_{c, c'}$ can be expressed as
\begin{align*}
X_{c, c'} = \frac{n}{x} \cdot \sum_{v \in V} X_{v, c, c'}  \text{.}
\end{align*}
We next compute the conditional expected value $\expected(X_{c}| N_{c,c'} = x)$:
\begin{align*}
\expected(X_{c, c'}| N_{c,c'} = x) &= \expected\left(\frac{n}{x} \cdot \sum_{v \in V} X_{v, c, c'} | N_{c, c'} = x \right)
                          = \frac{n}{x} \cdot \sum_{v \in V} \expected( X_{v, c, c'} | N_{c, c'} = x) \\
                          &= \frac{n}{x} \cdot \sum_{v \in V} \sum_{V'\subseteq V\colon |V'| = x}[v \in V']\cdot[c \succ_v c'] \\
                          &= \frac{n}{x} \cdot \sum_{v \in V} \frac{{n-1 \choose x-1}}{{n \choose x}} [c \succ_v c'] \\  
                          &= \sum_{v \in V}  [c \succ_v c'] = \sc_{\minimax}(c, c') \text{.}
\end{align*}

Next, we will use the Chernoff's inequality to upper-bound the probability that the value of random variable $X_{c, c'}$ does not differ from its expected value by a factor of~$\epsilon$. We first look at the conditional probability $\prob\Big(\left| X_{c, c'} - \expected(X_{c, c'}) \right| \geq \epsilon\expected(X_{c, c'}) | N_{c, c'} = x \Big)$. As in the proof of \Cref{thm:separable}, we note that the conditional variables $\{X_{v, c, c'} | N_{c, c'} = x\}_{v \in V}$ are not independent, yet they are all negatively correlated---the fact that one variable becomes 1 (resp., 0) can only decrease the probabilities that some other becomes 1 (resp., 0).
Thus, we can still apply the Chernoff's bound~\cite[Theorem 1.16, Corollary 1.10]{auger2011theory}, which states that for any negatively-correlated random variables~$X_1,\ldots,X_n \in [0,1]$ such that~$X = \sum_{i=1}^n X_i$ and any~$\delta \in [0,1]$ it holds that
\begin{align*}
\prob(|X - \expected(X)| \geq \delta \expected(X)) \leq 2\exp\left(-\delta^2\expected(X)/3\right).
\end{align*}
Since, for each $v \in V$, we have $X_{v, c, c'} \in \{0, 1\}$, and  $\expected\left(\sum_{v \in V} X_{v, c, c'} | N_{c, c'} = x \right) = \frac{x}{n}\sc_{\minimax}(c, c')$, we get that:
\begin{align*}
\prob\Big(\left| X_{c, c'} - \expected(X_{c, c'}) \right| \geq \epsilon\expected(X_{c, c'}) | N_{c, c'} = x \Big)  \leq 2\exp\left(- \frac{x\epsilon^2\sc_{\minimax}(c, c')}{3n}\right)
\end{align*}

Next, we get that:
\begin{align*}
&\prob\Big(\left| X_{c, c'} - \expected(X_{c, c'}) \right| \geq\epsilon\expected(X_{c, c'})\Big)  \\
&\qquad =\sum_{x = 0}^{n} \prob(N_{c, c'} = x) \cdot \prob\Big(\left| X_{c, c'} - \expected(X_{c, c'}) \right| \geq \epsilon\expected(X_{c, c'}) | N_{c, c'} = x \Big) \\
&\qquad \leq \sum_{x = 1}^{n} \prob(N_{c, c'} = x) \cdot 2\exp\left(- \frac{x\epsilon^2\sc_{\minimax}(c, c')}{3n}\right).
\end{align*}

Notice that~$\prob(N_{c, c'} = x)$ can be represented by the following.
First we decide for~$x$ out of~$n$ voters to rank~$c$ and~$c'$.
We then ask these~$x$ voters to rank~$l-2$ out of~$m-2$ remaining candidates and ask all other~$n-x$ voters to \emph{not} rank~$c$ \emph{and}~$c'$.
This can be modeled by~${m \choose \ell} - {m-2 \choose \ell -2}$ as~$m \choose \ell$ is the total number of possible sets to ask a voter to rank and as discussed before, there are~$m-2 \choose \ell -2$ sets that contain~$c$ and~$c'$.
Hence
\begin{align*}
 &\prob(N_{c, c'} = x) = {n \choose x} \cdot \frac{{{m-2 \choose \ell - 2}}^{x} \left( {m \choose \ell} - {m-2 \choose \ell -2} \right)^{n-x}}{{m \choose \ell}} \text{ and thus} \\
&\qquad \sum_{x = 1}^{n} \prob(N_{c, c'} = x) \cdot 2\exp\left(- \frac{x\epsilon^2\sc_{\minimax}(c, c')}{3n}\right) \\
&\qquad = \sum_{x = 1}^{n} {n \choose x} \cdot \frac{{{m-2 \choose \ell - 2}}^{x} \left( {m \choose \ell} - {m-2 \choose \ell -2} \right)^{n-x}}{{{m \choose \ell}}^n} \cdot 2\exp\left(- \frac{x\epsilon^2\sc_{\minimax}(c, c')}{3n}\right) \\
&\qquad = 2\exp\left(- \frac{\epsilon^2\sc_{\minimax}(c, c')}{3n}\right) \cdot \sum_{x = 1}^{n} {n \choose x} \cdot \left(\frac{\ell(\ell - 1)}{m(m-1)}\right)^x \cdot \left(1 - \frac{\ell(\ell - 1)}{m(m-1)}\right)^{n-x} \cdot \left(\frac{1}{e}\right)^x.
\end{align*}

Again using the binomial identity~$(x+y)^n = \sum_{k=1}^{n} {n \choose k}x^k y^{n-k}$, we get
\begin{align*}
&2\exp\left(- \frac{\epsilon^2\sc_{\minimax}(c, c')}{3n}\right) \cdot \sum_{x = 1}^{n} {n \choose x} \cdot \left(\frac{\ell(\ell - 1)}{m(m-1)}\right)^x \cdot \left(1 - \frac{\ell(\ell - 1)}{m(m-1)}\right)^{n-x} \cdot \left(\frac{1}{e}\right)^x\\
&\qquad = 2\exp\left(- \frac{\epsilon^2\sc_{\minimax}(c, c')}{3n}\right) \cdot \left(1 - \frac{\ell(\ell - 1)}{m(m-1)} + \frac{\ell(\ell - 1)}{em(m-1)}\right)^{n}\\
&\qquad = 2\exp\left(- \frac{\epsilon^2\sc_{\minimax}(c, c')}{3n}\right) \cdot \left(1 - \frac{\ell(\ell - 1)(e-1)}{m(m-1)e} \right)^{n} \\
&\qquad \leq 2\exp\left(- \frac{\epsilon^2\sc_{\minimax}(c, c')}{3n}\right) \cdot \left(1 - \frac{\ell^2}{2m^2} \right)^{n} \\
&\qquad \leq 2\exp\left(- \frac{\epsilon^2\sc_{\minimax}(c, c')}{3n}\right) \cdot \exp\left(- \frac{\ell^2}{2m^2}n \right) \\
&\qquad \leq 2\exp\left(- \frac{\epsilon^2 \ell^2 \sc_{\minimax}(c, c')}{6m^2}\right) 
\end{align*}

Finally, let $c_{\min} = \argmin_{c' \neq c} \sc_{\minimax}(c, c')$. The probability that for candidate $c$ the score computed by \Cref{alg:minimax} differs from its true Minimax score by a multiplicative factor of at least $1 \pm \epsilon$ is upper-bounded by
$$ \prob\Big(\min_{c' \neq c}X_{c, c'}  < \sc_{\minimax}(c, c_{\min})(1 - \epsilon) \Big) + \prob\Big(\min_{c' \neq c}X_{c, c'}  > \sc_{\minimax}(c, c_{\min})(1 + \epsilon) \Big).$$
Clearly, we have:
\begin{align*}
\prob\Big(\min_{c' \neq c}X_{c, c'}  > \sc_{\minimax}(c, c_{\min})(1 + \epsilon) \Big) &\leq \prob\Big(X_{c, c_{\min}}  > \sc_{\minimax}(c, c_{\min})(1 + \epsilon) \Big) \\
&\leq \exp\left(- \frac{\epsilon^2 \ell^2 \sc_{\minimax}(c, c_{\min})}{6m^2}\right) \text{.}
\end{align*}
Further, since by definition for each $c'$ we have~$\sc_{\minimax}(c, c') \geq \sc_{\minimax}(c, c_{\min})$, it holds that:
\begin{align*}
\prob\Big(\min_{c' \neq c}X_{c, c'}  < \sc_{\minimax}(c, c_{\min})(1 - \epsilon) \Big) &\leq \sum_{c' \neq c}\prob\Big(X_{c, c'}  < \sc_{\minimax}(c, c_{\min})(1 - \epsilon) \Big)\\
&\leq \sum_{c' \neq c}\prob\Big(X_{c, c'}  < \sc_{\minimax}(c, c')(1 - \epsilon) \Big)\\
&\leq \sum_{c' \neq c}\exp\left(- \frac{\epsilon^2 \ell^2 \sc_{\minimax}(c, c')}{6m^2}\right) \\
&\leq \sum_{c' \neq c}\exp\left(- \frac{\epsilon^2 \ell^2 \sc_{\minimax}(c, c_{\min})}{6m^2}\right) \\
&\leq (m-1) \exp\left(- \frac{\epsilon^2 \ell^2 \sc_{\minimax}(c, c_{\min})}{6m^2}\right) \text{.}
\end{align*}
Thus we get that
\begin{align*}
&\prob\Big(\min_{c' \neq c}X_{c, c'}  < \sc_{\minimax}(c, c_{\min})(1 - \epsilon) \Big) + \prob\Big(\min_{c' \neq c}X_{c, c'}  > \sc_{\minimax}(c, c_{\min})(1 + \epsilon) \Big) \\
&\leq m\exp\left(- \frac{\epsilon^2 \ell^2 \sc_{\minimax}(c, c_{\min})}{6m^2}\right)\text{, completing the proof.} \qedhere
\end{align*}
\end{proof}

\section{Deterministic Approach ($\ell$-Truncated Elections)}\label{sec:deterministic}
When not asking each voter about each candidate, one always has to decide whether each voter is asked about random candidates or about specific ones.
On the one hand, asking about specific positions in preference rankings, allows one to focus on the top ones that seem to contain more relevant information; especially when the goal is to select the winner, who---intuitively---is more likely to appear in top positions. On the other hand, asking voters about random candidates might be more advantageous as the input may contain dependencies between candidates that are not known a priori.

In this section we investigate the case when each voter is asked about her~$\ell$ most preferred candidates.
We will then describe an algorithm that is guaranteed to approximate the true score at least as good as any other algorithm and analyze its performance for Borda and for Minimax.
We will then show a general lower bound on the approximation ratio, that is, we show that no algorithm can approximate the true score in the worst case arbitrarily good and see that it matches the bound for Borda and almost matches the bound we computed for Minimax.

\subsection{The Best Approximation Algorithm for $\ell$-Truncated Elections}\label{sec:_determnistic_algorithm_pos}
Let us start by describing the algorithm that for each $\ell$-truncated instance gives the best possible approximation guarantee, that is, the best approximation of the true winner in the worst-case full preference profile that induces the given $\ell$-truncated instance.
We mention that the idea of this algorithm is very similar to the one behind the algorithms for minimizing the maximal regret~\cite{lu-bou:maximizing-regret}, yet the analysis of the approximation ratio of the algorithm is new to this paper.

Consider an election $E$ and let $E_{\ell}$ be the $\ell$-truncated instance obtained from $E$.
Observe that when given $E_{\ell}$ and choosing a winner, the worst case occurs if the picked winner is ranked at the very last position by all voters that did not rank this candidate in $E$ among the first~$\ell$ positions, and the true winner (that our algorithm did not pick) is ranked at position~$\ell + 1$ by each voter who did not rank this candidate.

For each candidate $c$ we compute two scores: The worst possible score that $c$ is guaranteed to get (denoted by $\worst(c)$)---this score is obtained when $c$ is ranked last whenever it is not among the top-$\ell$ positions---and the best possible score that $c$ can get (denoted by $\best(c)$)---the score obtained by ranking~$c$ at position~$\ell+1$ whenever it is not ranked among the first~$\ell$ positions.
Let~$a,b_1,$ and~$b_2$ be the candidates with the highest~$\worst$, the highest~$\best$, and the second highest~$\best$ score, respectively.
If any candidate~$c \neq b_1$ is declared winner, then we can guarantee an approximation ratio of~$\nicefrac{\worst(c)}{\best(b_1)}$, which is clearly maximized by~$c = a$.

If candidate~$b_1$ is declared winner, then we can guarantee an approximation ratio of~$\nicefrac{\worst(b_1)}{\best(b_2)}$.
Thus, an optimal approximation ratio is achieved by an algorithm that computes all the possible scores and then checks whether~$a$ or~$b_1$ guarantees a better result.
If $\nicefrac{\worst(a)}{\best(b_1)} \geq \nicefrac{\worst(b_1)}{\best(b_2)}$ it declares $a$ the winner and otherwise it picks $b_1$.

We first show that these two different cases (sometimes choosing $a$ and other times choosing~$b_1$ guaranteeing a better approximation ratio) can both occur.
Afterwards we analyze the guarantees given by each of the two cases and conclude with a family of instances that prove that the obtained results are tight.

\begin{example}
 Consider the following two instances, each with~$5$~voters~$\{v_1,v_2,v_3,v_4,v_5\}$, $6$~candidates~$\{a,b,c,d,e,f\}$, and~$\ell = 3$:
 \begin{align*}
 I_1: && \quad I_2: \\
 v_1 &: b > e > a &  \quad v_1 &: b > c > a \\
 v_2 &: b > e > a &  \quad v_2 &: b > d > a \\
 v_3 &: b > d > a &  \quad v_3 &: c > e > a \\
 v_4 &: c > e > a &  \quad v_4 &: e > a > b \\
 v_5 &: d > a > c &  \quad v_5 &: d > e > a \text{.}
\end{align*}
Assume our goal is to find the candidate that would best approximate the Borda winner. The scores are presented in the two tables below:
\begin{center}
\begin{tabular}{c|c|c}
$I_1$ & $\worst$ & $\best$\\\hline
$a$ & $16$ & $16$\\
$b$ & $15$ & $19$\\
$c$ & $8$ & $14$\\
$d$ & $9$ & $15$\\
$e$ & $12$ & $16$\\
$f$ & $0$ & $10$
\end{tabular}\hspace{1cm}
\begin{tabular}{c|c|c}
$I_2$ & $\worst$ & $\best$\\\hline
$a$ & $16$ & $16$\\
$b$ & $13$ & $17$\\
$c$ & $9$ & $15$\\
$d$ & $9$ & $15$\\
$e$ & $13$ & $17$\\
$f$ & $0$ & $10$
\end{tabular}.
\end{center}
\vspace{.2cm}

In both instances~$a$ is the candidate with the highest~$\worst$ score and~$b$ is (among) the candidates with highest~$\best$ score.
In instance~$I_1$, it is best to declare~$b$ winner as it guarantees an approximation ratio of~$\nicefrac{15}{16} > \nicefrac{16}{19}$.
In instance~$I_2$ on the other hand, it is best to declare~$a$ winner since it guarantees an approximation ratio of~$\nicefrac{16}{17} > \nicefrac{13}{16}$.
Notice that the second term in each inequality is the approximation ratio guaranteed by choosing the respective other candidate.
\end{example}

Interestingly, while our algorithm provides the best possible approximation, it can select a candidate that is not a possible winner, i.e., that is not a winner in any profile consistent with the truncated ballot at hand.

\begin{example}
	Consider the following instance with~$5$ candidates~$a,b,c,d,e$, four voters~$v_1,v_2,\ldots, v_4$, and~$\lambda = (3,1,1,1,0)$.
	\begin{align*}
		v_1\colon b \succ a\\
		v_2\colon c \succ a\\
		v_3\colon d \succ a\\
		v_4\colon e \succ a
	\end{align*}
	It holds that~$\worst(a) = 4$,~$\worst(b) = \worst(c) = \worst(d) = \worst(e) = 3$,~$\best(a)=4$ and~$\best(b) = \best(c) = \best(d) = \best(e) = 6$ and therefore declaring~$a$ winning achieves the best approximation ratio.
	However, in any election that is consistent with the given truncated election, at least two candidates in~$\{b,c,d,e\}$ get at least 5 points while~$a$ always gets 4 points.
	Thus,~$a$ is not a possible winner.
\end{example}

\subsection{Positional Scoring Rules: Approximation Guarantees for $\ell$-Truncated Elections}

In this section we continue our analysis of the algorithm from \Cref{sec:_determnistic_algorithm_pos}, focusing on how well it approximates positional scoring rules having only access to $\ell$-truncated elections.
We will now prove guarantees that each of the two rules (choosing the candidate with highest~$\worst$ respectively~$\best$ score) provide.

\begin{theorem}\label{thm:positional_general_res}
Let $\calR$ be a positional scoring rule defined by the scoring function $\lambda(i) = \alpha_i$. The algorithm from \Cref{sec:_determnistic_algorithm_pos} for $\ell$-truncated elections gives an approximation guarantee of
\begin{align*}
\frac{\sum_{i=1}^{\ell} \alpha_i}{m \alpha_{\ell+1} + \frac{\alpha_1 - \alpha_{\ell+1}}{\alpha_1} \sum_{i=1}^{\ell} \alpha_i} \text{.}
\end{align*}
\end{theorem}
\begin{proof}
Let us first assume that the algorithm picks the candidate~$a$ with the highest~$\worst$ score as a winner. The average~$\worst$ score each candidate gets is~$\avg_w = \nicefrac{n}{m} \cdot \sum_{i=1}^{\ell} \alpha_i$. Thus, $\worst(a) \geq \nicefrac{n}{m} \cdot \sum_{i=1}^{\ell} \alpha_i$. Let $b$ be the candidate, different from $a$ that has maximal~$\best$ score. Clearly, $\worst(b) \leq \worst(a)$. Further, if we fix $\worst(b)$, then $\best(b)$ is maximized when candidate~$b$ is ranked among the top $\ell$ positions (the positions that are counted for $\worst(b)$) as few times as possible. This way the number of voters who do not rank $b$ is maximized---and these are the voters who can contribute additional score (apart from $\worst(b)$) to $\best(b)$. That is, if $\worst(b)$ is fixed, then $\best(b)$ is higher when $b$ is ranked high by fewer voters rather than when it is ranked lower but by more voters. Consequently, we can lower bound $\best(b)$ by:
\begin{align*}
\worst(b) + \left(n - \frac{\worst(b)}{\alpha_1}\right)\alpha_{\ell+1}
\end{align*}
Thus, the approximation ratio in this case is at least
\begin{align*}
\mathrm{approx}_1 &\geq \frac{\worst(a)}{\best(b)} \geq \frac{\worst(a)}{\worst(b) + (n - \frac{\worst(b)}{\alpha_1}) \alpha_{\ell+1}} \\
&\geq \frac{\worst(b)}{\worst(b) + (n - \frac{\worst(b)}{\alpha_1}) \alpha_{\ell+1}} \geq \frac{\avg_w}{\avg_w + (n - \frac{\avg_w}{\alpha_1}) \alpha_{\ell+1}} \\
&\geq \frac{\sum_{i=1}^{\ell} \alpha_i}{m \alpha_{\ell+1} + \frac{\alpha_1 - \alpha_{\ell + 1}}{\alpha_1} \sum_{i=1}^{\ell} \alpha_i} \text{.}
\end{align*}

\medskip
We next turn to the case when the algorithm picks the candidate $b$ with maximum $\best$ score. The average best score of a candidate is given by
\begin{align*}
{\avg}_b = \frac{n}{m} \cdot ((m-\ell)\alpha_{\ell + 1} + \sum_{i=1}^{\ell}\alpha_i) \text{.}
\end{align*}
Clearly, $\best(b) \geq \avg_b$. With a fixed $\best(b)$ the $\worst$ score of $b$ is minimized when $b$ is ranked by as few voters as possible (the reasoning is similar as in the previous case). If $b$ is ranked first by~$x$~voters, then its $\best$ score would be $\alpha_1 x + (n-x)\alpha_{\ell+1}$. By solving:
\begin{align*}
\alpha_1 x + (n-x)\alpha_{\ell+1} = \best(b),
\end{align*}
we get that $x = \frac{\best(b) - n \alpha_{\ell+1}}{\alpha_1 - \alpha_{\ell+1}}$. Thus, $b$ gets the $\worst$ score of at least $\alpha_1\frac{\best(b) - n \alpha_{\ell+1}}{\alpha_1 - \alpha_{\ell+1}}$.
Consequently, we can lower-bound the approximation ratio by:
\begin{align*}
\mathrm{approx}_2 &\geq \frac{\worst(b)}{\best(b)} \geq \frac{\alpha_1\frac{\best(b) - n \alpha_{\ell+1}}{(\alpha_1 - \alpha_{\ell+1})}}{\best(b)} \geq \frac{\alpha_1\frac{{\avg}_b - n \alpha_{\ell+1}}{(\alpha_1 - \alpha_{\ell+1})}}{{\avg}_b} \\
&= \frac{\sum_{i=1}^{\ell}\alpha_i - \ell \alpha_{\ell+1}}{((m-\ell)\alpha_{\ell + 1} + \sum_{i=1}^{\ell}\alpha_i)\cdot \frac{\alpha_1 - \alpha_{\ell+1}}{\alpha_1}} \text{.}
\end{align*}

\noindent
It is easy to verify that for all~$a,b,c,d$ with~$0 \leq a \leq c \leq d$,~$b \geq 0$, and~$\frac{(d-c)a}{c} \geq b$ it holds that~$\frac{c}{d} \geq \frac{c-a}{d-a-b}$.
Substituting into this~$a = \ell \alpha_{\ell+1}$,~$b = (m-\ell) \frac{\alpha_{\ell+1}^2}{\alpha_1}$,~$c = \sum_{i=1}^{\ell}\alpha_i$ and~$d=m\alpha_{\ell+1} + \frac{\alpha_1-\alpha_{\ell+1}}{\alpha_1}\sum_{i=1}^{\ell}\alpha_i$, we get that
\begin{align*}
\frac{\sum_{i=1}^{\ell} \alpha_i}{m \alpha_{\ell+1} + \frac{\alpha_1 - \alpha_{\ell+1}}{\alpha_1} \sum_{i=1}^{\ell} \alpha_i} >  \frac{\sum_{i=1}^{\ell}\alpha_i - \ell \alpha_{\ell+1}}{((m-\ell)\alpha_{\ell + 1} + \sum_{i=1}^{\ell}\alpha_i)\cdot \frac{\alpha_1 - \alpha_{\ell+1}}{\alpha_1}} \text{.}
\end{align*}
We will only show that~$c \leq d$ and~$\frac{(d-c)a}{c} \geq b$ as everything else is trivial.
Observe that
\begin{align*}
c \leq c + \alpha_{\ell+1} (m-\ell) \leq c + m \alpha_{\ell+1} - \alpha_{\ell+1} \sum_{i=1}^{\ell} \frac{\alpha_i}{\alpha_1}=d  \text{ and}\\
\frac{(d-c)a}{c} = \frac{m - \sum_{i=1}^{\ell} \frac{\alpha_i}{\alpha_1}}{\sum_{i=1}^{\ell} \alpha_i} \ell \alpha_{\ell+1}^2 \geq \frac{m - \ell}{\ell \alpha_1} \ell \alpha_{\ell+1}^2 = b.
\end{align*}
Since the algorithm always picks the value that results in a higher ratio, we get the thesis. 
\end{proof}

\Cref{thm:positional_general_res} gives a very general result that applies to any positional scoring rule. For instance, for $k$-approval we get the approximation ratio of $\nicefrac{\ell}{m}$.

\begin{corollary}
The algorithm from \Cref{sec:_determnistic_algorithm_pos} for $k$-approval with $\ell$-truncated elections, $k > \ell$, gives the approximation guarantee of $\nicefrac{\ell}{m}$.
\end{corollary}
\begin{proof}
We instantiate the expression from \Cref{thm:positional_general_res} for $k$-approval:
\begin{align*}
\frac{\sum_{i=1}^{\ell} \alpha_i}{m \alpha_{\ell+1} + \frac{\alpha_1 - \alpha_{\ell+1}}{\alpha_1} \sum_{i=1}^{\ell} \alpha_i}  = \frac{\ell}{m + 0} \text{.}
\end{align*}
\end{proof}

For the Borda rule, we get the approximation of $\frac{\ell}{m + \frac{\ell}{m - 1} \cdot \ell}$ which on the plot looks similarly to~$\frac{\ell}{m}$ (see the left-hand side plot in~\Cref{fig:borda_approx_deterministic}).

\begin{corollary}
The algorithm from \Cref{sec:_determnistic_algorithm_pos} for Borda with $\ell$-truncated elections, gives the approximation guarantee of $\frac{\ell}{m + \frac{\ell}{m - 1} \cdot \ell} $.
\end{corollary}
\begin{proof}
The approximation ratio follows from \Cref{thm:positional_general_res}:
\begin{align*} 
\frac{\sum_{i=1}^{\ell} \alpha_i}{m \alpha_{\ell+1} + \frac{\alpha_1 - \alpha_{\ell+1}}{\alpha_1} \sum_{i=1}^{\ell} \alpha_i} &=
\frac{\frac{(2m-\ell -1)\ell}{2} }{m (m - \ell -2) + \frac{\ell}{m - 1} \cdot \frac{(2m-\ell -1)\ell}{2}} \\
&= \frac{\frac{\ell}{2} }{m \frac{(m - \ell -2)}{(2m-\ell -1)} + \frac{\ell}{m - 1} \cdot \frac{\ell}{2}} \geq \frac{\frac{\ell}{2} }{\frac{m}{2} + \frac{\ell}{m - 1} \cdot \frac{\ell}{2}} 
= \frac{\ell}{m + \frac{\ell}{m - 1} \cdot \ell} \text{.}
\end{align*}
\end{proof}

\begin{figure}[!t]
  \centering
  \minipage{0.48\textwidth}
    \includegraphics[width=\linewidth]{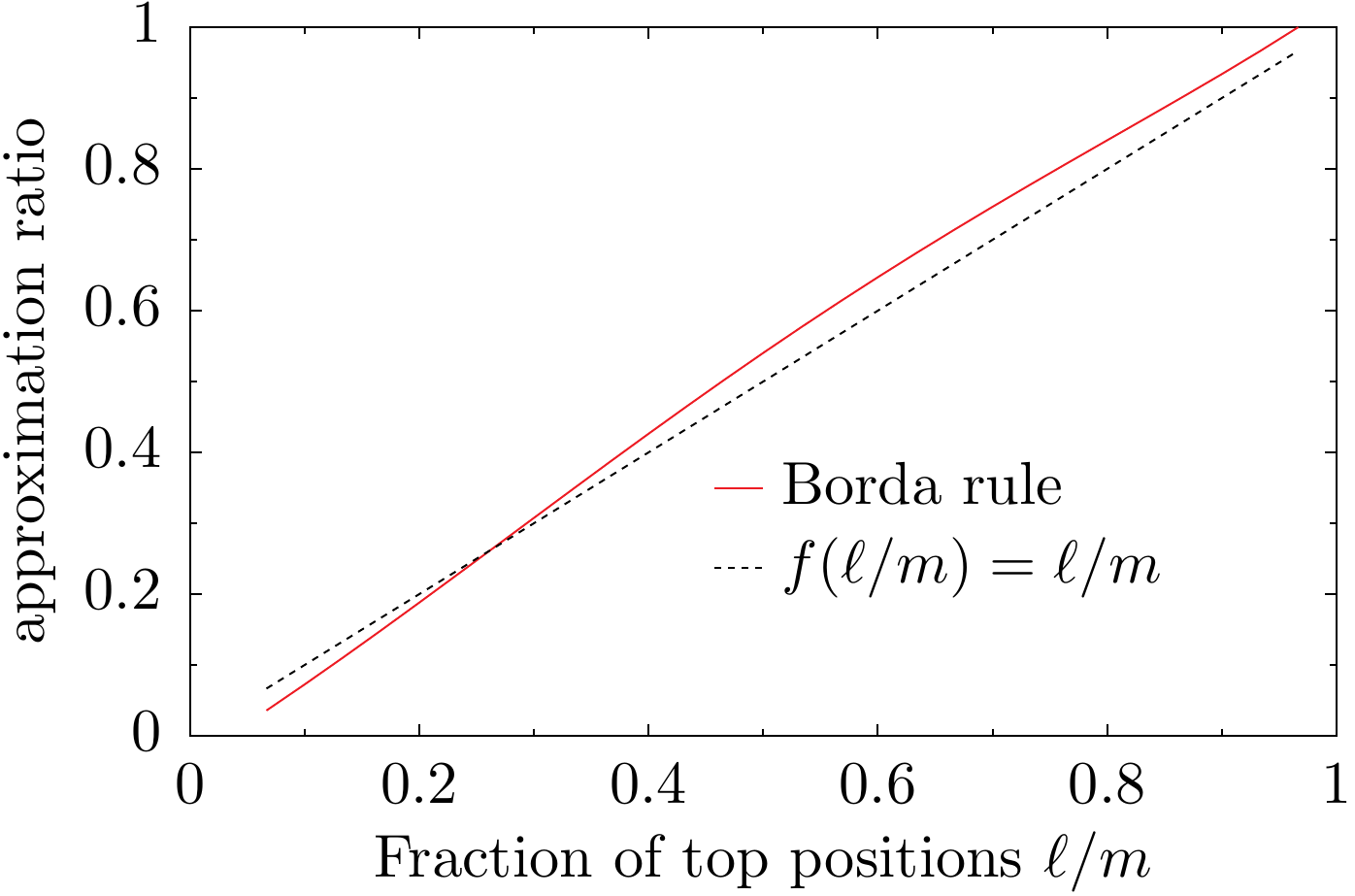}
  \endminipage\hfill
  \minipage{0.48\textwidth}
    \includegraphics[width=\linewidth]{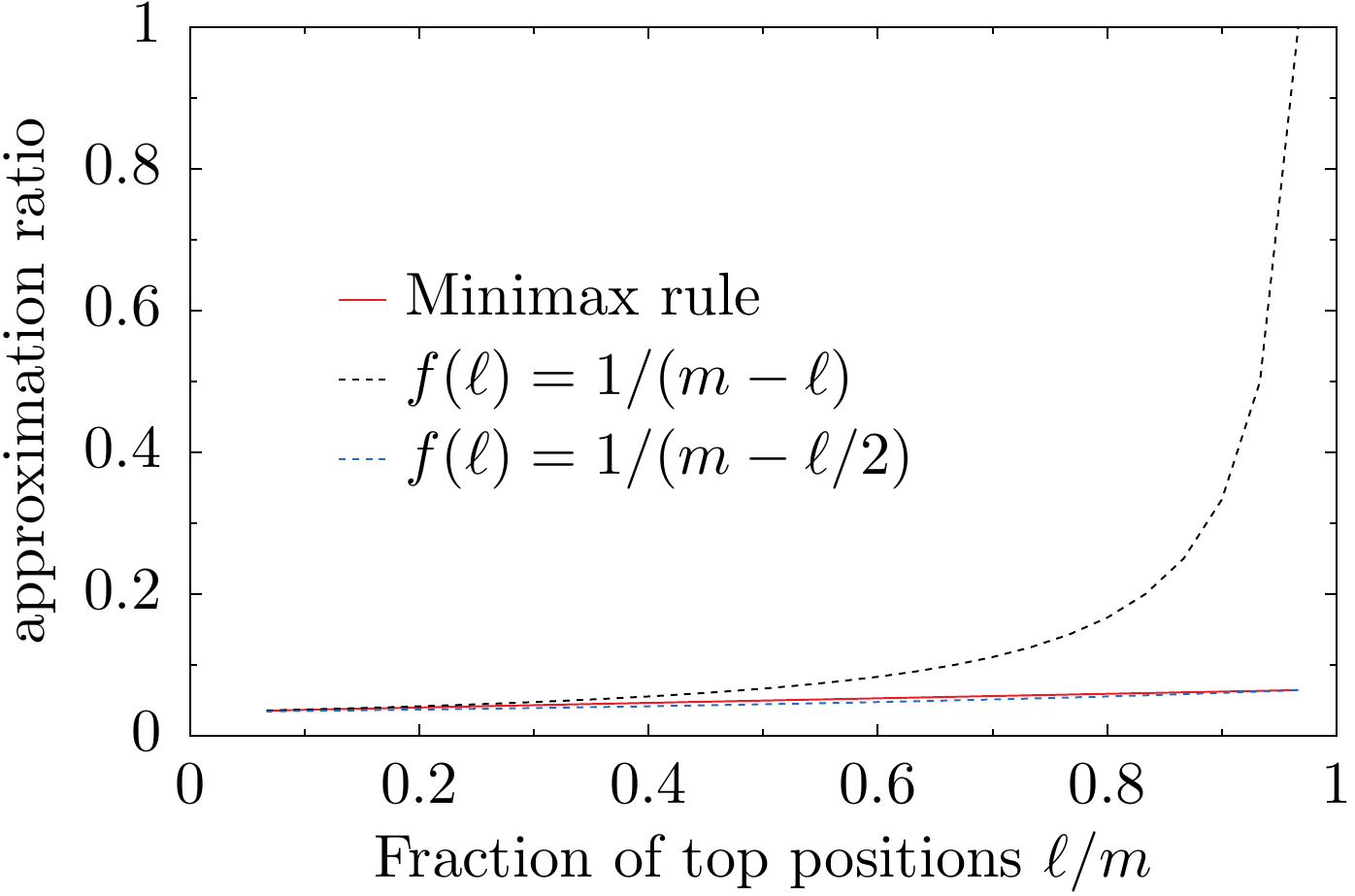}
  \endminipage\hfill
\caption{The approximation ratios for the deterministic algorithms from \Cref{sec:_determnistic_algorithm_pos} for the Borda rule (the left plot) and for the Minimax rule (the right plot). The plot was generated for $m = 30$. Note that the plotted line for the Minimax rule is almost indistinguishable from the plotted line for~$1 / (m-\ell/2)$}\label{fig:borda_approx_deterministic}
\end{figure}

We conclude by providing intuitive explanation of instances that match the bound from \Cref{thm:positional_general_res}.
In these instances all candidates get roughly the same~$\worst$ score and there are two candidates~$a,b$ that also get average~$\worst$ score but only appear as few times as possible in the first~$\ell$ positions (in the first or second positions only).
If candidate~$a$ is declared the winner by any rule then she gets~$0$ points from all voters, that did not rank her in the first two positions and~$b$ gets~$m -\ell -1$ points from these voters.
Otherwise the winning candidate gets~$0$ points from all voters that did not rank her in the first~$\ell$ positions and~$a$ gets~$m-\ell-1$ points whenever she is not ranked first or second.
Notice that no rule can distinguish between the different instances we just constructed and therefore building the instance after the rule picked a winner is permitted.
In either case, the candidate that is declared winner by any rule gets points equal to the average~$\worst$ score and the ``true winner''~$a$ or~$b$ gets~$\avg_{\worst} + (n - \frac{\avg_{\worst}}{m-1}) (m - \ell - 1)$ points.
Observe that this gives exactly the same conditions as for the computation of~$\mathrm{approx}_1$ in the proof of \Cref{thm:positional_general_res} and hence we have a matching upper bound.

\subsection{Minimax Rule}
Let us now move to the analysis of the Minimax rule.
We start by showing that no deterministic algorithm for Minimax can guarantee a better approximation ratio than~$\frac{1}{m-\ell}$.

\begin{theorem}
	\label{thm:mmlb}
	There exists no rule for $\ell$-truncated elections $\calF$ that is a $\frac{1}{m-\ell-\varepsilon}$-approximation of the Minimax rule for any~$\varepsilon > 0$.
\end{theorem}

\begin{proof}
	Consider the following~$\ell$-truncated instance of election:
	Let~$C = \{c_0,c_1,\ldots,c_{m-1}\}$ be the set of candidates and let~$V = \{v_0,v_1,\ldots v_{m-1}\}$ be the set of voters.
	Let the truncated preference list of voter~$v_i$~be 
        \begin{align*}
        c_i \succ c_{i+1\m m} \succ c_{i+2\m m} \succ \dots \succ c_{i+\ell-1\m m} \text{.}
        \end{align*}
	Due to symmetry, any candidate can be declared winner by each algorithm.
	For the sake of simplicity, let us assume that candidate~$c_2$ is declared winner.
	We can then complement this instance by inserting all candidates that were not ranked by some voter~$v_i$ in the order suggested by the subscript, that is,
	\begin{align*}
		c_1 \succ c_2 \succ \dots \succ c_m.
    \end{align*}
	On the one hand, only voter~$v_2$ prefers~$c_2$ over~$c_1$ and all other~$m-1$ voters prefer~$c_1$ over~$c_2$. Thus, the Minimax score of $c_2$ is 1. 
	On the other hand, the Minimax score of $c_1$ is $m-\ell$. Indeed, the strongest contender to~$c_1$ is~$c_m$ and~$\ell$ voters prefer~$c_m$ over~$c_1$ (for~$\ell < m$) while~$m-\ell$ voters prefer~$c_1$ over~$c_m$.
	Thus, no (deterministic) algorithm can achieve a better approximation ration than~$ \nicefrac{\sc_{\minimax}(c_2)}{\sc_{\minimax}(c_1)} = \nicefrac{1}{m-\ell}$.
\end{proof}

Note that in the construction in \Cref{thm:mmlb} one can increase the number of voters to be much larger than the number of candidates by simply copying all voters a sufficient number of times. 

\Cref{thm:mmlb} already shows that with $\ell$-trucnated ballots Minimax cannot be well approximated. In particular, the bound for the Minimax rule is much worse than for scoring-based rules. 
We do not know whether the bound from \Cref{thm:mmlb} is tight. Yet, we can show that a simplified variant of the algorithm from \Cref{sec:_determnistic_algorithm_pos} that computes the maximum~$\worst$ score of each candidate and declares the one with the highest score winner, achieves an approximation ratio of~$\frac{1}{(m-\ell) \cdot \left(1 + \frac{\ell^2}{m^2 - \ell^2 - m + \ell}\right)}$.
This approximation ratio is lower-bounded by~$\frac{1}{m - \nicefrac{\ell}{2}}$, which means that for reasonably small $\ell$ it (almost) matches the upper bound from~\cref{thm:mmlb} (see the right-hand side plot in \Cref{fig:borda_approx_deterministic} for the comparison of these two bounds).

\begin{theorem}
	\label{thm:mmub}
	The algorithm from \Cref{sec:_determnistic_algorithm_pos} approximates the Minimax rule in $\ell$-truncated elections within a factor of
\begin{align*}
\frac{1}{(m-\ell) \cdot \left(1 + \frac{\ell^2}{m^2 - \ell^2 - m + \ell}\right)} \geq \frac{1}{m-\nicefrac{\ell}{2}} \text{.}
\end{align*}
\end{theorem}

\begin{proof}
	We will prove our claim by showing that if all candidates have a~$\worst$ score of at most~$x$, then all candidates have a maximum~$\best$ score of at most~$x \cdot (m-\ell) \cdot \left(1 + \frac{\ell^2}{m^2 - \ell^2 - m + \ell}\right)$.

	Assume towards a contradiction that all candidates have a~$\worst$ score of at most~$x$ and there exists a candidate~$f$ with~$\best(f) > x \cdot (m-\ell) \cdot \left(1 + \frac{\ell^2}{m^2 - \ell^2 - m + \ell}\right)$.
	We say that $c \succ_v^t d$ for voter~$v$ and candidates~$c$ and~$d$ if and only if~$c$ is preferred over~$d$ by voter~$v$ in the~$\ell$-truncated instance, that is, either~$c$ and~$d$ are both among the~$\ell$ most preferred candidates of~$v$ and~$c$ is preferred over~$d$ or only~$c$ is among the first~$\ell$ candidates.
	Since $$\best(f) = \min_{c' \in C \setminus \{f\}} \{ n - |\{v \in V \mid c' \succ_v^t f\}| \}$$ and~$\best(f) > x \cdot (m-\ell) \cdot \left(1 + \frac{\ell^2}{m^2 - \ell^2 - m + \ell}\right)$, it follows that for all~$c \in C \setminus \{f\}$ we have
\begin{align*}
 n - |\{v \in V \mid c \succ_v^t f\}| > x \cdot (m-\ell) \cdot \left(1 + \frac{\ell^2}{m^2 - \ell^2 - m + \ell}\right).
\end{align*}
	Hence for all~$c \neq f$ it holds that 
\begin{equation}
|\{v \in V \mid c \succ_v^t f\}| < n - x \cdot (m-\ell) \cdot \left(1 + \frac{\ell^2}{m^2 - \ell^2 - m + \ell}\right).\label{eq:pref_set_bound}
\end{equation}
	We now analyze~$\worst(f)$.
	Let~$\occ(d)$ be the number of times a candidate~$d$ \emph{occurs} in the truncated instance, that is, the number of voters that rank candidate~$d$ among the first~$\ell$ positions.
	First, observe that for each candidate~$c$ it holds that $$\worst(c) \geq \min_{c' \in C\setminus \{c\}} \{\occ(c) - |\{v \in V \mid c' \succ_v^t c\}|\}$$ as~$c$ is preferred over any~$c'$ at least that often.
	Hence,
\begin{align*}
\worst(f) &\geq \min_{c \in C \setminus\{f\}} \{\occ(f) - |\{v \in V \mid c \succ_v^t f\}|\} \\
          &\overset{\cref{eq:pref_set_bound}}> \occ(f) - (n - x \cdot (m-\ell) \cdot \left(1 + \frac{\ell^2}{m^2 - \ell^2 - m + \ell}\right).
\end{align*}
	Since, by assumption,~$\worst(f) \leq x$, it holds that~
\begin{align*}
x > \occ(f) - n + x \cdot (m-\ell) \cdot \left(1 + \frac{\ell^2}{m^2 - \ell^2 - m + \ell}\right),
\end{align*}
or, equivalently,
\begin{equation}	
\occ(f) < n - x \cdot (m-\ell) \cdot \left(1 + \frac{\ell^2}{m^2 - \ell^2 - m + \ell}\right)+x. \label{eq:occ_bound}
\end{equation}
	Second, notice that if~$f$ is ranked among the top~$\ell$ candidates by at most~$\occ(f)$ voters, then there are~$n - \occ(f)$ voters that do not rank~$f$ among the first~$\ell$ positions and by pigeonhole principle there is a candidate~$c \in C \setminus \{f\}$ with~$|\{v \mid c \succ_v^t f\}| \geq (n-\occ(f)) \nicefrac{\ell}{m-1}$.
	As discussed above, from \Cref{eq:pref_set_bound} it follows that
\begin{align*}
(n-\occ(f)) \nicefrac{\ell}{m-1} < n - x \cdot (m-\ell) \cdot \left(1 + \frac{\ell^2}{m^2 - \ell^2 - m + \ell}\right),
\end{align*}
 which is equivalent to
\begin{align*}
	\occ(f) > x \cdot \frac{(m-1)\cdot (m-\ell) \cdot \left(1 + \frac{\ell^2}{m^2 - \ell^2 - m + \ell}\right)}{\ell} - \frac{n(m-1)}{\ell} + n.
\end{align*}
	Plugging in~\Cref{eq:occ_bound} into this inequality, we get that
	\begin{align*}
	&x \cdot \frac{(m-1) (m-\ell) \left(1 + \frac{\ell^2}{m^2 - \ell^2 - m + \ell}\right)}{\ell} - \frac{n(m-1)}{\ell} +n \\
        & < n - x \cdot (m-\ell) \cdot \left(1 + \frac{\ell^2}{m^2 - \ell^2 - m + \ell}\right)+x\\
	\iff \quad &x \cdot \left[ \frac{(m + \ell -1)(m-\ell)\left(1 + \frac{\ell^2}{m^2 - \ell^2 - m + \ell}\right) - \ell}{\ell}\right ] < n \cdot \frac{m-1}{\ell}\\
	\iff \quad &x < \frac{n \cdot (m-1)}{(m + \ell -1)(m-\ell)\left(1 + \frac{\ell^2}{m^2 - \ell^2 - m + \ell}\right) - \ell} \\
                   &  = \frac{n (m-1)}{(m+\ell-1)(m-\ell)+ \ell^2 - \ell}\\
	\quad & = \frac{n(m - 1)}{m^2 - m} = \frac{n}{m} \text{.}
	\end{align*}

	Notice that on the other hand~$x \geq \nicefrac{n}{m}$ as by pigeonhole principle there is a candidate that is ranked first at least~$\nicefrac{n}{m}$ times and hence has a $\worst$ score of at least~$\nicefrac{n}{m}$.
	Thus, we have reached a contradiction, completing the first part of the proof. We finish the proof by proving $$(m-\ell) \cdot \left(1 + \frac{\ell^2}{m^2 - \ell^2 - m + \ell}\right) \leq m-\nicefrac{\ell}{2}.$$ Observe that
\begin{align*}
(m-\ell) \cdot \left(1 + \frac{\ell^2}{m^2 - \ell^2 - m + \ell}\right) &= m - \ell + \frac{\ell^2 (m-\ell)}{(m-\ell)(m+\ell-1)}\\
&= m-\ell+\frac{\ell^2}{m+\ell-1}\\
&\overset{m > \ell}\leq  m-\ell+\frac{\ell^2}{2\ell}\\
&=  m-\nicefrac{\ell}{2} \text{.}\qedhere
\end{align*}
\end{proof}

\section{Experimental Evaluation}\label{sec:experiments}

In \Cref{sec:randomized,sec:deterministic} we have assessed the worst-case guarantees of our approximation algorithms. In this section we investigate how these guarantees depend on particular distributions of the the voters' preferences.
We tested the following distributions over preference rankings:
\begin{description}
\item[Impartial Culture (IC).] Under the Impartial Culture model each ranking over the candidates is equally probable.
\item[One-dimensional Euclidean Model (1D).] First, we associate each voter and each candidate with a point from the interval $[0, 1]$---these points are sampled independently and uniformly at random. Then, each voter ranks the candidates according to her distance, preferring the ones which are closer to those which are farther.
\item[Mixture of Mallows' Models (MMM).] In the Mallows' model~\cite{mal:j:mallows} we are given a reference ranking $\pi$ and a real value $\phi \in [0,1]$; the probability of sampling a ranking $\tau$ is proportional to $\phi^{d_K(\pi, \tau)}$, where $d_K(\pi, \tau)$ is the number of swaps of adjacent candidates that are required to turn $\phi$ into $\tau$. We used a mixture of three Mallows' models: for each of the three models we drawn the reference ranking $\pi$ and the real value $\phi$ uniformly at random. Next, we sampled the parameters $\lambda_1, \lambda_2, \lambda_3$ that sum up to one; to generate a ranking we first pick one of the three models, we pick the $i$-th model with probability $\lambda_i$, and we generate the ranking according to the Mallows' model we picked.
\item[Single Peaked Impartial Culture (SPIC).] In order to generate a profile we first randomly select a reference ranking. Then, we generate rankings that are single-peaked with respect to the reference ranking. Each such single peaked ranking is equally probable. For a definition and discussion on single-peaked preferences we refer the reader to the book chapter by Elkind et al.~\cite{ELP-trends}. 
\end{description}

For each distributions~$\calD$ over preferences and for each approximation algorithm $\calA$ we ran computer simulations as follows: We set the number of candidates to $m=50$ and tested for $\ell \in \{2,5,8\}$. We ran simulations for the number of voters $n$ ranging from $10$ to $1000$ in steps of $25$. For each combination of values of $(\ell, n)$ we ran 500 independent experiments, each time computing the ratio~$r(\calA, \calD)$ between the score of the candidate returned by algorithm $\calA$ to the score of the optimal candidate. The averages of these ratios (averaged over the aforementioned 500 simulations) and the corresponding standard deviations for the Borda and the Minimax rules are depicted in \Cref{fig:approx_exper_borda} and \Cref{fig:approx_exper_minimax}, respectively. 

\begin{figure}[!thb]
\begin{center}
Impartial Culture

\minipage{0.42\textwidth}
  \includegraphics[width=\linewidth]{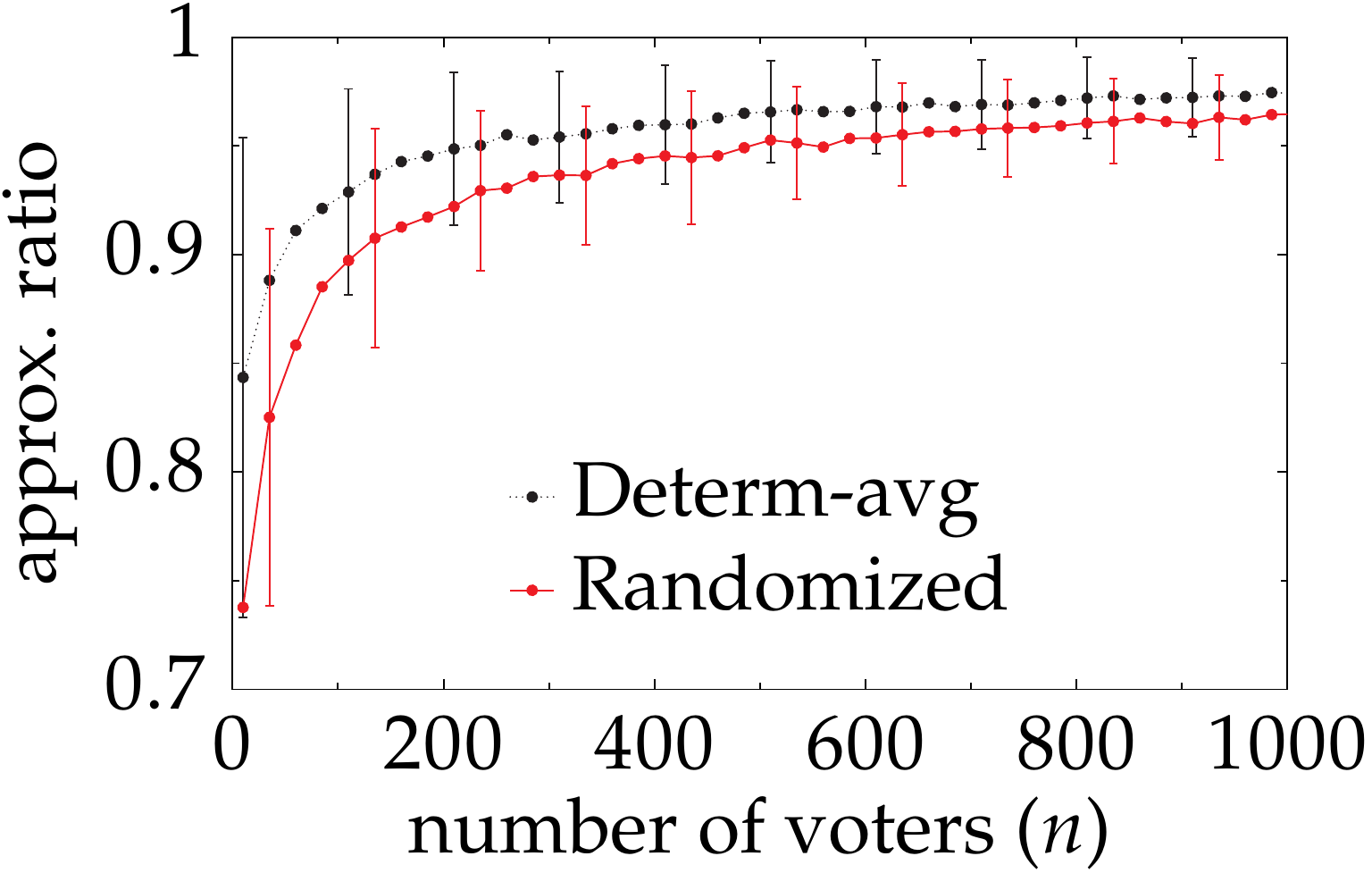}
  (a) $\ell = 2$
\endminipage\hfill
\minipage{0.42\textwidth}
  \includegraphics[width=\linewidth]{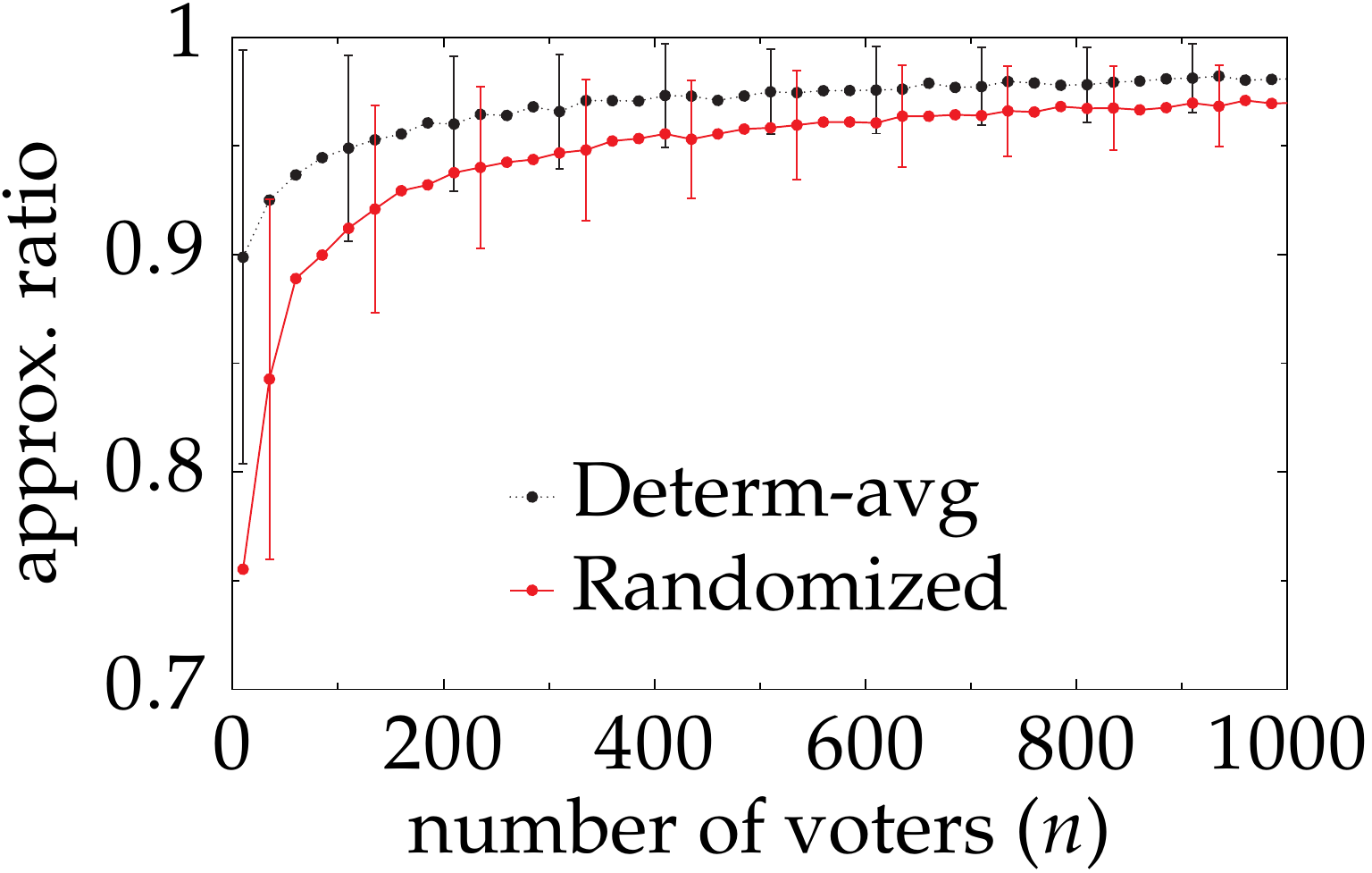}
  (b) $\ell = 5$
\endminipage

One-dimensional Euclidean Model

\minipage{0.42\textwidth}
  \includegraphics[width=\linewidth]{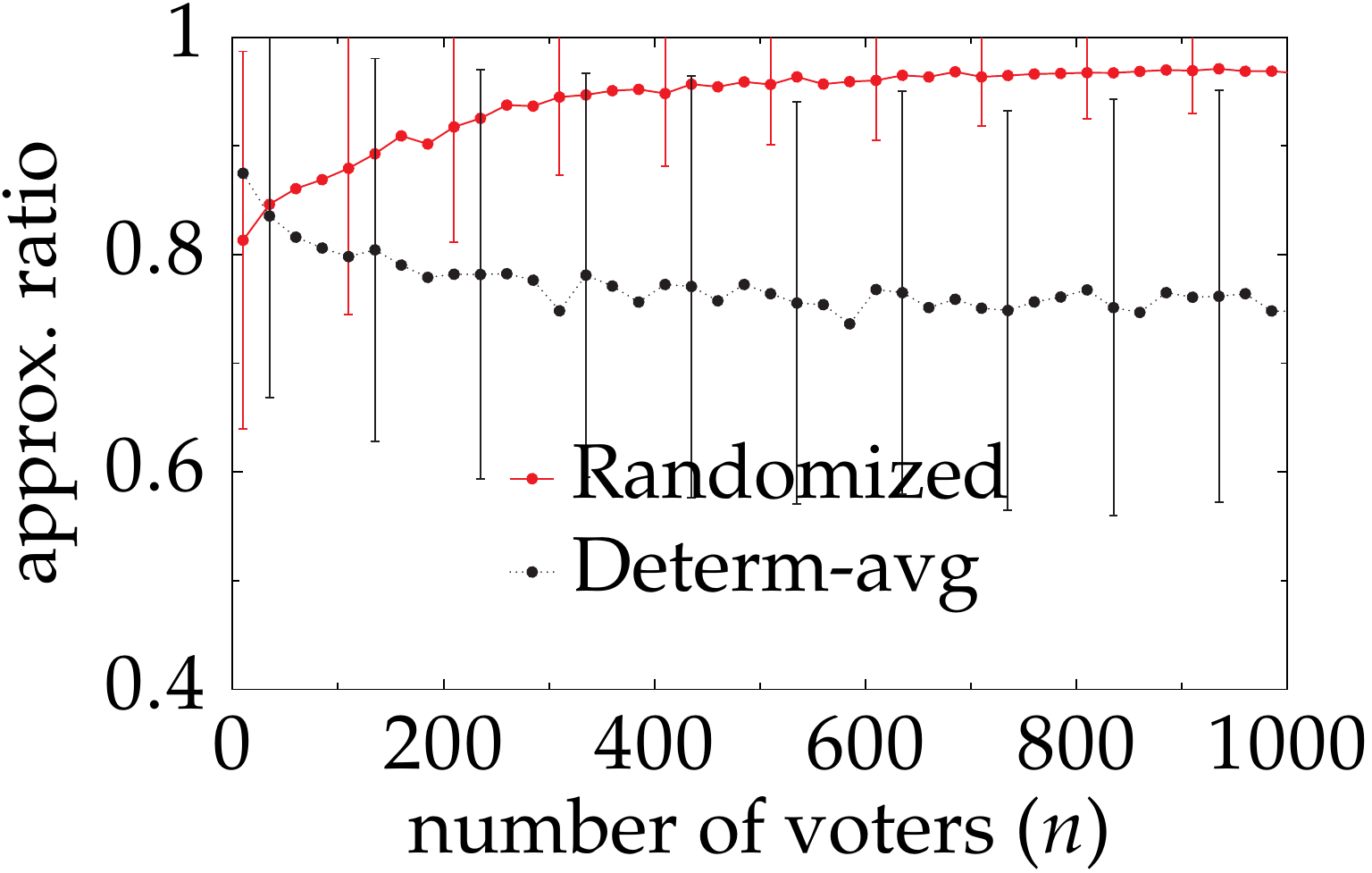}
  (a) $\ell = 2$
\endminipage\hfill
\minipage{0.42\textwidth}
  \includegraphics[width=\linewidth]{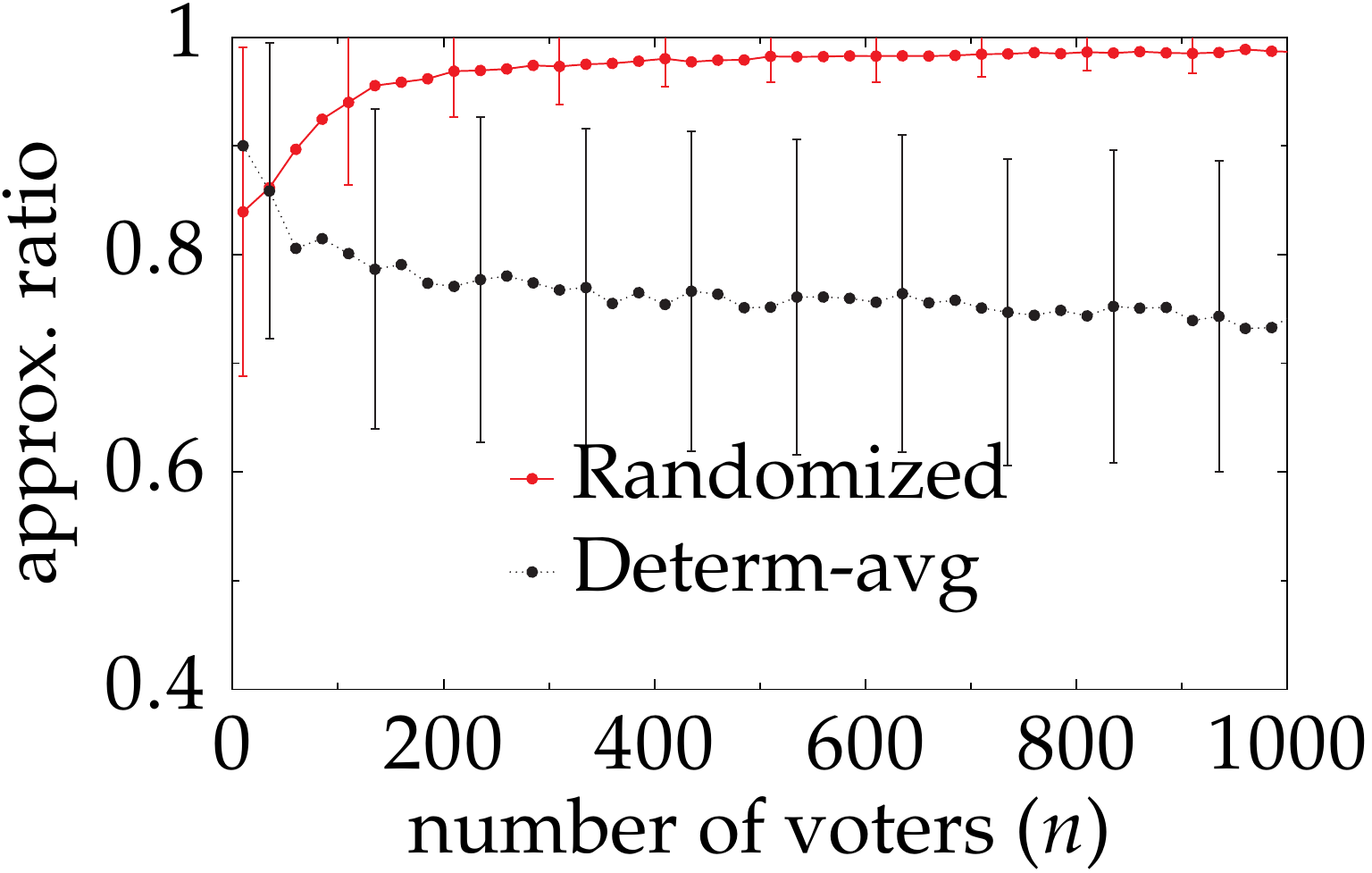}
  (b) $\ell = 5$
\endminipage

Mixture of Mallows' Models

\minipage{0.42\textwidth}
  \includegraphics[width=\linewidth]{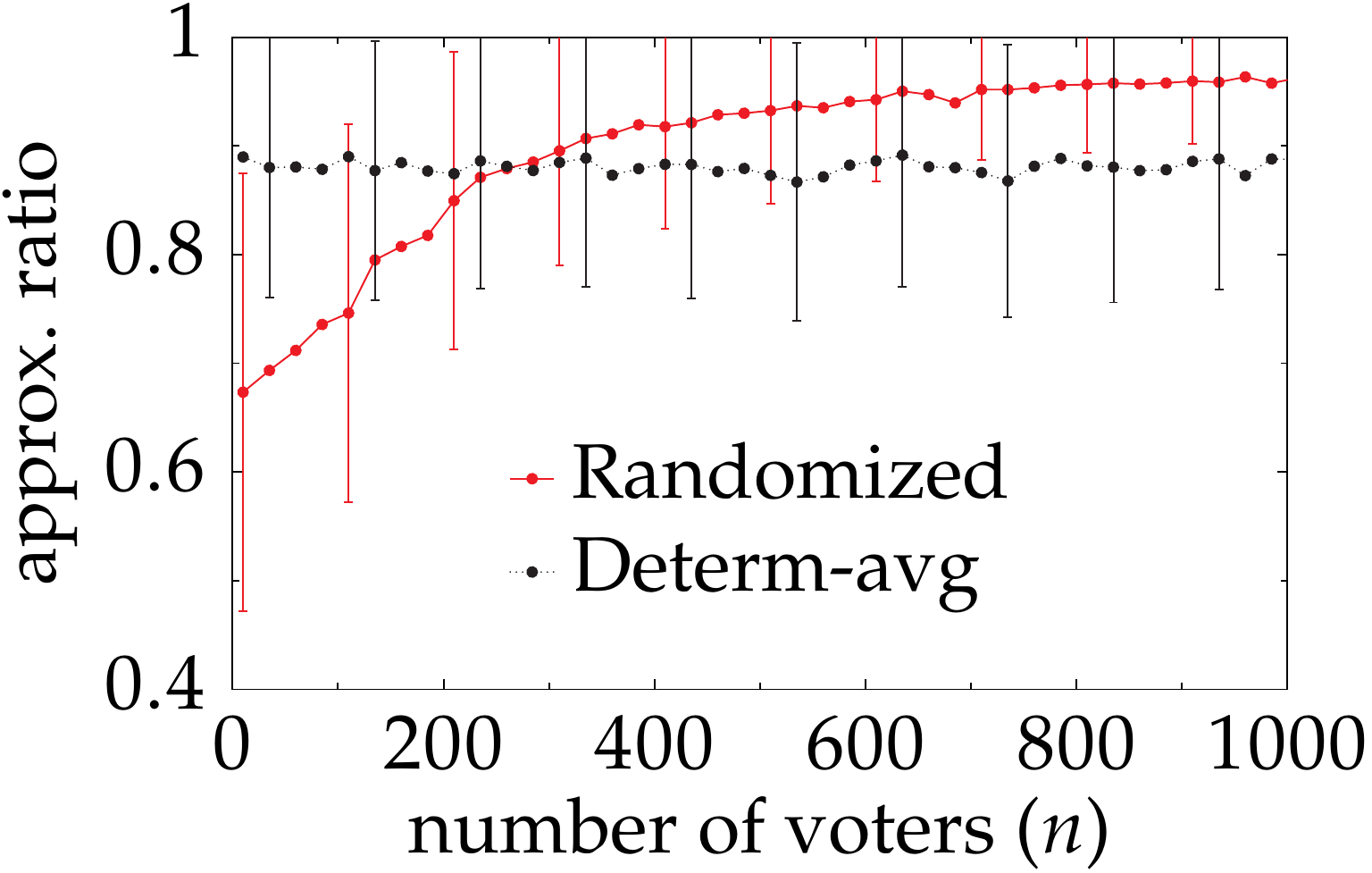}
  (a) $\ell = 2$
\endminipage\hfill
\minipage{0.42\textwidth}
  \includegraphics[width=\linewidth]{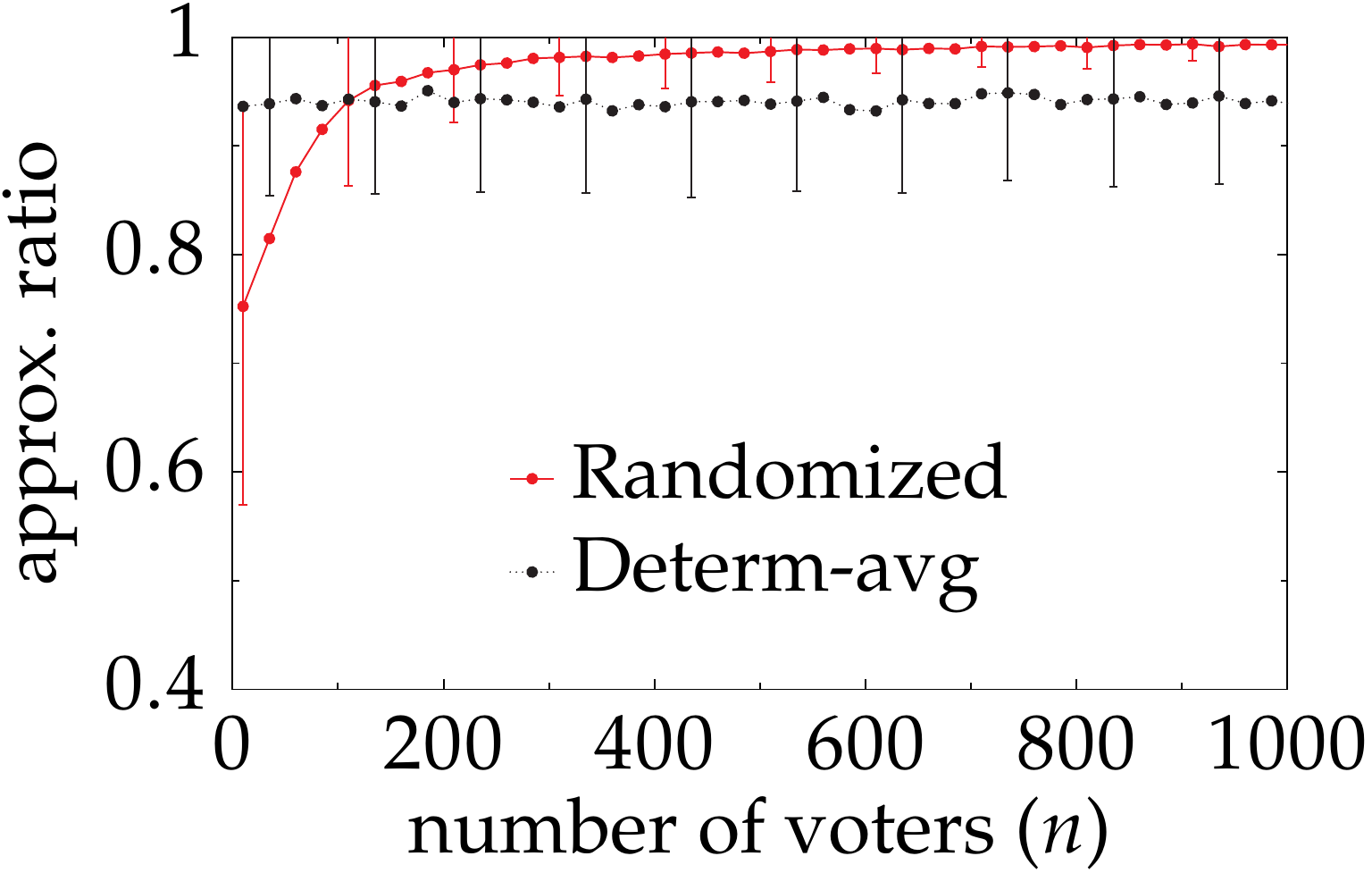}
  (b) $\ell = 5$
\endminipage

Single Peaked Impartial Culture

\minipage{0.42\textwidth}
  \includegraphics[width=\linewidth]{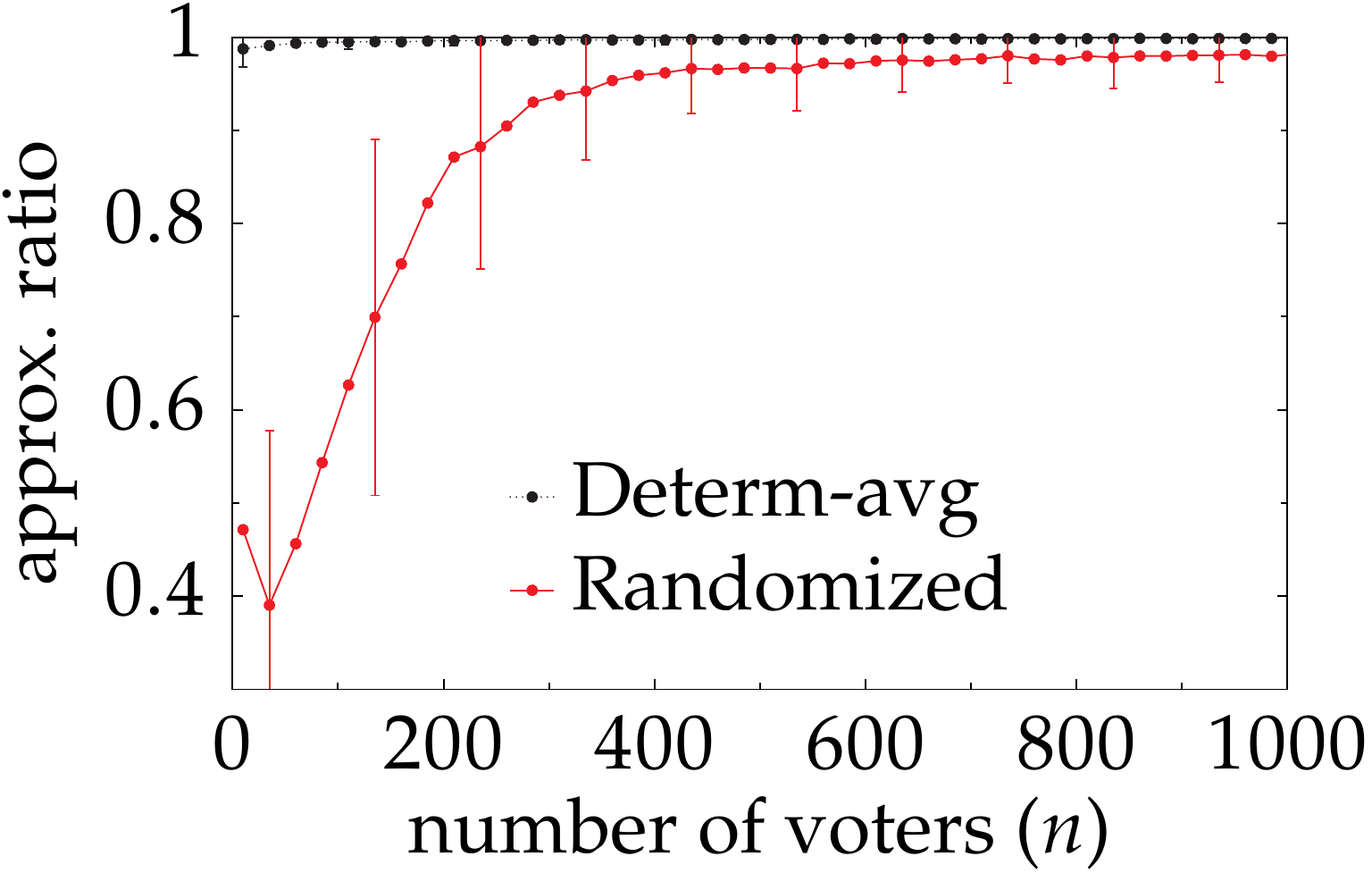}
  (a) $\ell = 2$
\endminipage\hfill
\minipage{0.42\textwidth}
  \includegraphics[width=\linewidth]{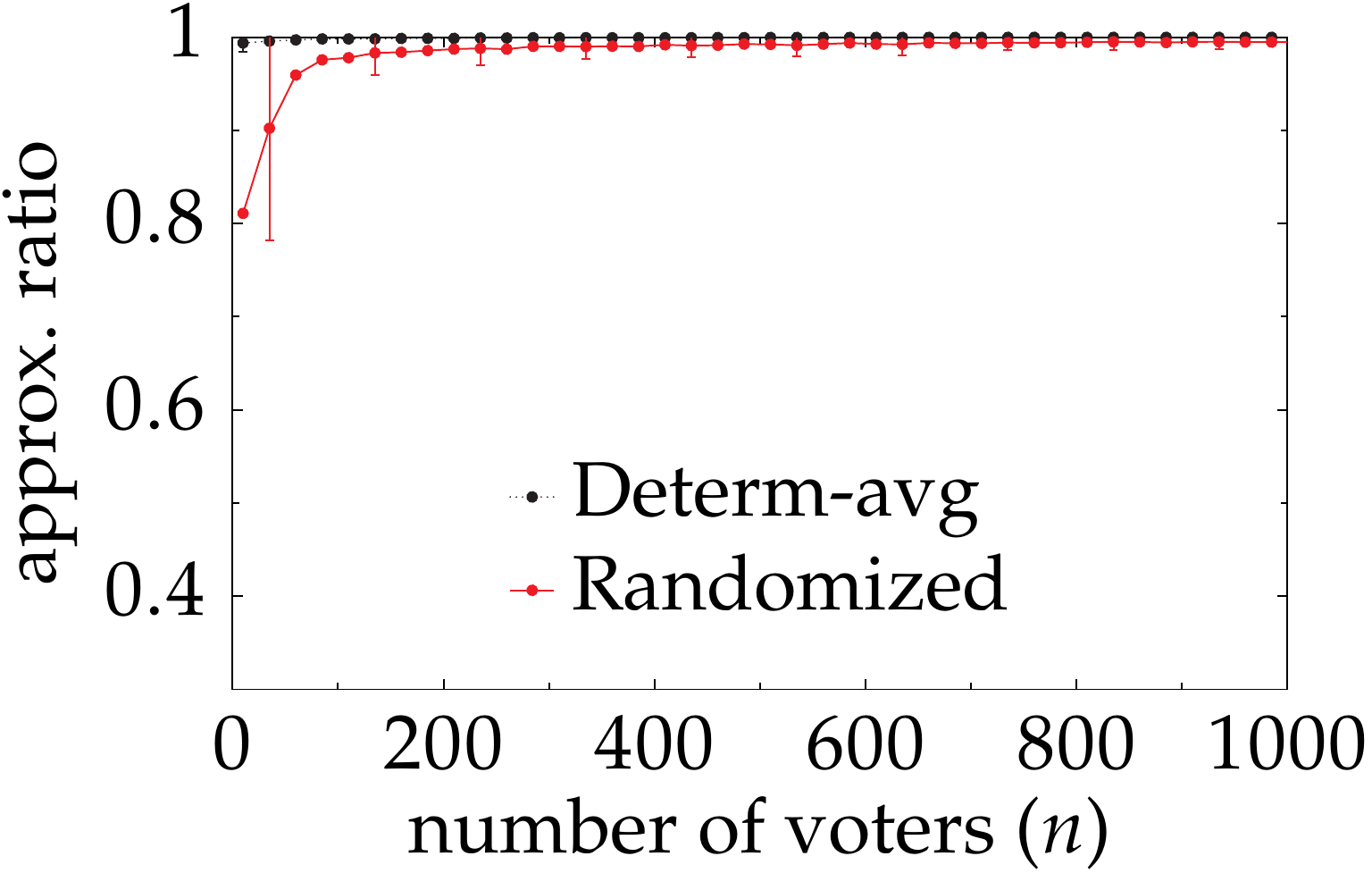}
  (b) $\ell = 5$
\endminipage
\end{center}
\caption{Approximation ratio for the two algorithms for the Borda rule (the randomized algorithm asks to compare $\ell$ random candidates, and the deterministic asks for the $\ell$-truncated ballot) assessed through computer simulations.}\label{fig:approx_exper_borda}
\end{figure}

\subsection{Approximation Algorithms for the Borda Rule}\label{sec:borda_experiments}

We empirically tested how well the two algorithms that we analyzed theoretically in the previous sections approximate the Borda rule. Specifically, we implemented \Cref{alg:psf}---which we will refer to as \textsc{Randomized}, and the algorithm described in \Cref{sec:_determnistic_algorithm_pos}. We also checked two other deterministic heuristics, that appear simple and intuitive:
\begin{enumerate}
\item The variant of the deterministic algorithm from \Cref{sec:_determnistic_algorithm_pos} that always picks the candidate with the highest $\worst$ score.
\item An algorithm we call \textsc{Deter-avg} that, for each voter $v_i$ and candidate $c_j$ assigns to $c_j$ the score
\begin{enumerate}
\item $\beta(\pos_i(c_j))$ if $\pos_i(c_j) \leq \ell$,
\item the average score of the unranked positions $\sum_{p=\ell+1^m}\beta(p) / (m-\ell)$, otherwise.
\end{enumerate}
Then, the algorithm picks the candidate with the highest total score.
\end{enumerate}
The three deterministic algorithms were almost indistinguishable in our simulations---\textsc{Deter-avg} was slightly better than the other two. Thus, for readability we present the results only for \textsc{Deter-avg} and \textsc{Randomized} and omit the description of the results for the other two deterministic algorithms. We found the following:
\begin{enumerate}
\item For preferences with no or with little structure, such as those generated by IC and SPIC, the deterministic algorithm gives better results. For preferences with more structure, e.g., those obtained from 1D and MMM models, the randomized algorithm significantly outperforms the deterministic ones.
\item For each preference distribution that we tested the randomized algorithm gives high quality approximations unless the number of voters is very small. Our results suggest to ask each voter to rank a random subsets of alternatives when the goal is to approximate the Borda rule with limited information from each voter and the number of voters exceeds a couple of hundreds.
\end{enumerate}

\subsection{Approximation Algorithms for the Minimax Rule}

Similarly to \Cref{sec:borda_experiments}, we empirically tested how well the randomized algorithm (\Cref{alg:minimax}) and the deterministic algorithm from \Cref{sec:_determnistic_algorithm_pos} approximate the Minimax rule. We refer to the two algorithms as \textsc{Randomized} and \textsc{Deterministic}, respectively. We also tested two other natural heuristics. For each two candidates $c$ and $c'$, let $n(c,c')$ denote the number of voters who (i) rank $c$ and $c'$ among their~$\ell$ most preferred candidates and prefer~$c$ over~$c'$ or (ii) who rank $c$ but not~$c'$ among their top $\ell$ positions. Then:
\begin{enumerate}
\item In our first heuristic algorithm, for each pair of candidates, $c$ and $c'$, we use a method similar to Minimax, but we replace $\sc_{\minimax}(c, c')$ by~$n(c, c')$. Then, similarly as in the case of the original Minimax rule we compute for each candidate $c$ the score $\min_{c' \neq c}n(c, c')$ and pick the candidate $w$ with the maximal score.
\item In the second heuristic, we set replace~$\sc_{\minimax}(c, c')$ by 
\begin{align*}
n \cdot \frac{n(c, c')}{n(c, c') + n(c', c)} \text{.}
\end{align*}
\end{enumerate} 

In our simulation \textsc{Deterministic} outperformed the two heuristic algorithms we mentioned above, hence we present our results only for \textsc{Deterministic} and \textsc{Randomized}. We observed the following:
\begin{enumerate}
\item The randomized algorithm for the Minimax rule needs to ask each voter to be compare more candidates than in case of Borda to achieve a good approximation. For $m=50$ candidates, asking each voter to compare $\ell=8$ of them already gave good results for sufficiently many voters.
\item The deterministic algorithm usually performs better than the randomized one, yet there are distributions (e.g., the one-dimensional Euclidean model) where the quality of winners returned by the deterministic algorithm is much worse than those returned by the randomized algorithm. On the other hand, for each distribution that we tested, the randomized algorithm consistently was giving good results when the number of voters and the number of candidates to ask each voter to rank were sufficiently large.  
\end{enumerate}

\begin{figure}[!thb]
\begin{center}

Impartial Culture

\minipage{0.42\textwidth}
  \includegraphics[width=\linewidth]{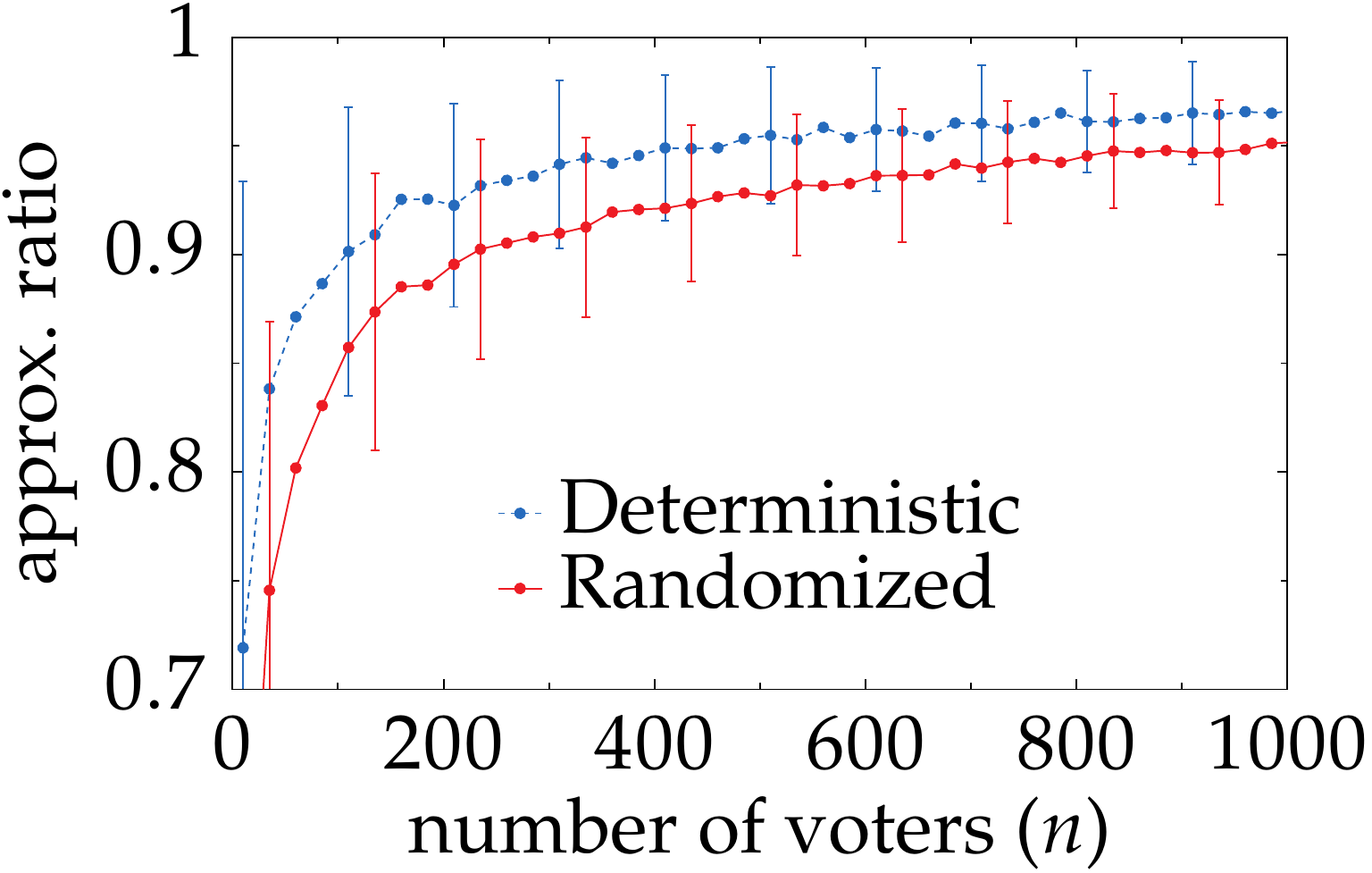}
  (a) $\ell = 2$
\endminipage\hfill
\minipage{0.42\textwidth}
  \includegraphics[width=\linewidth]{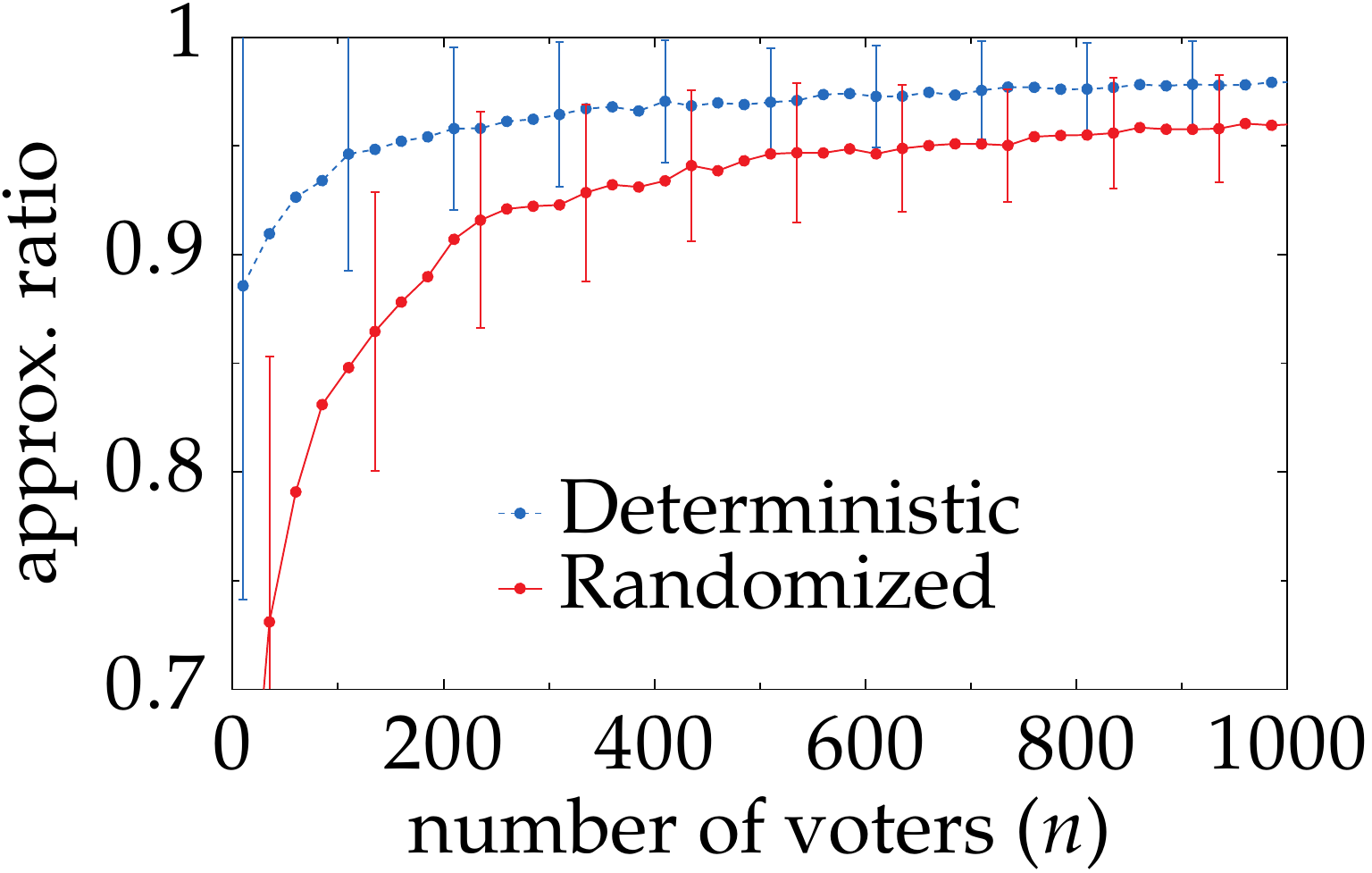}
  (b) $\ell = 8$
\endminipage

One-dimensional Euclidean Model

\minipage{0.42\textwidth}
  \includegraphics[width=\linewidth]{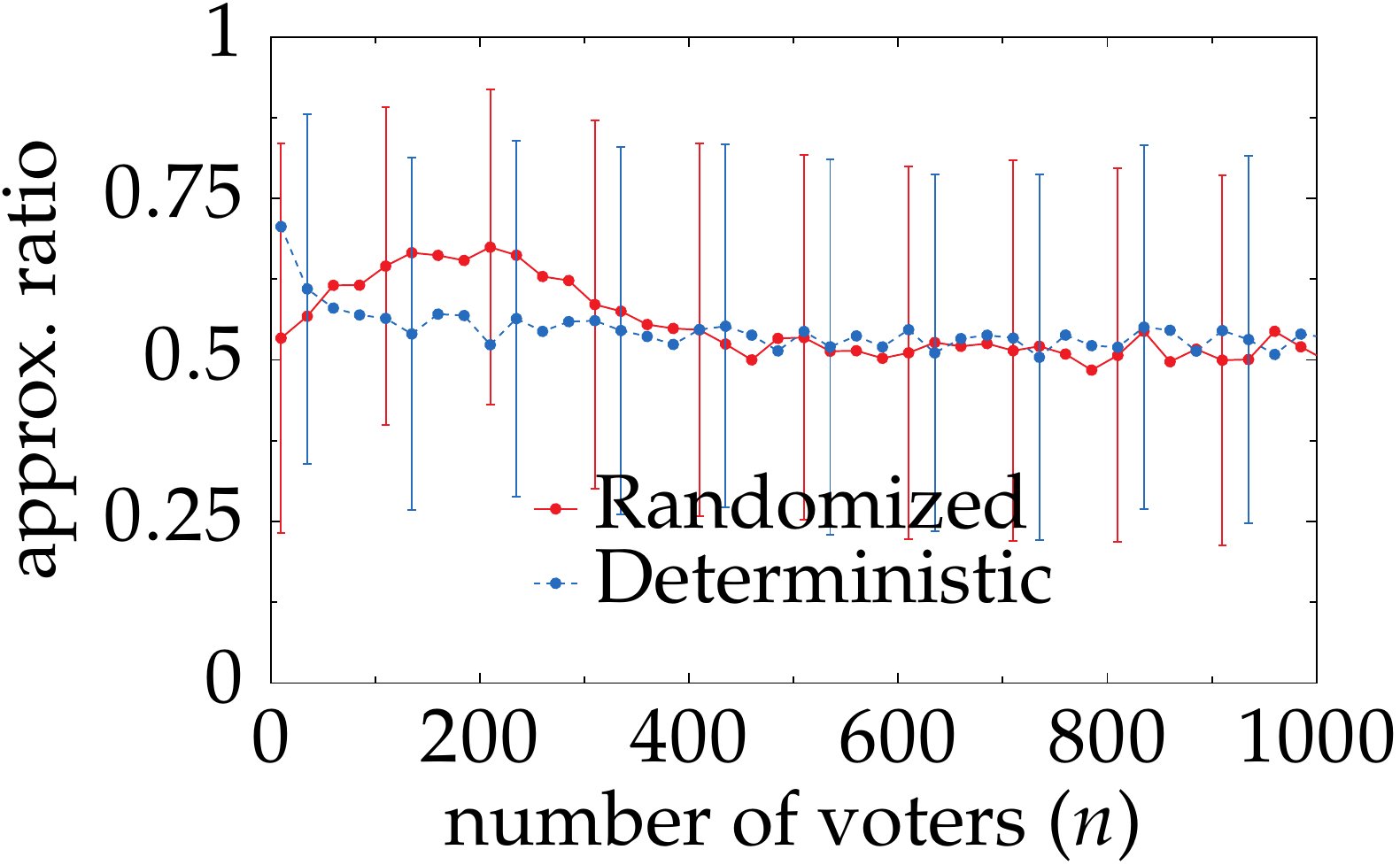}
  (a) $\ell = 2$
\endminipage\hfill
\minipage{0.42\textwidth}
  \includegraphics[width=\linewidth]{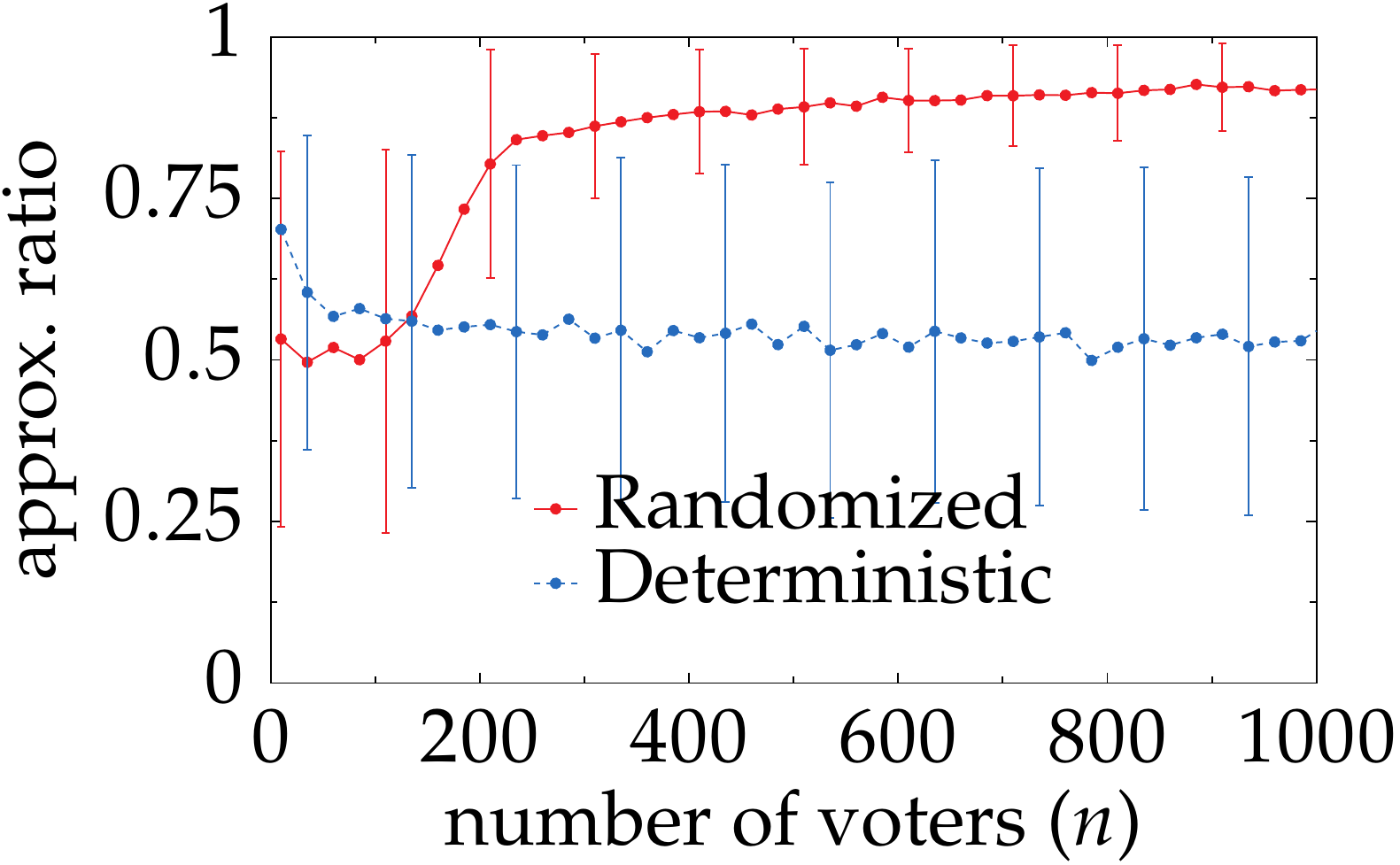}
  (b) $\ell = 8$
\endminipage

Mixture of Mallows' Models

\minipage{0.42\textwidth}
  \includegraphics[width=\linewidth]{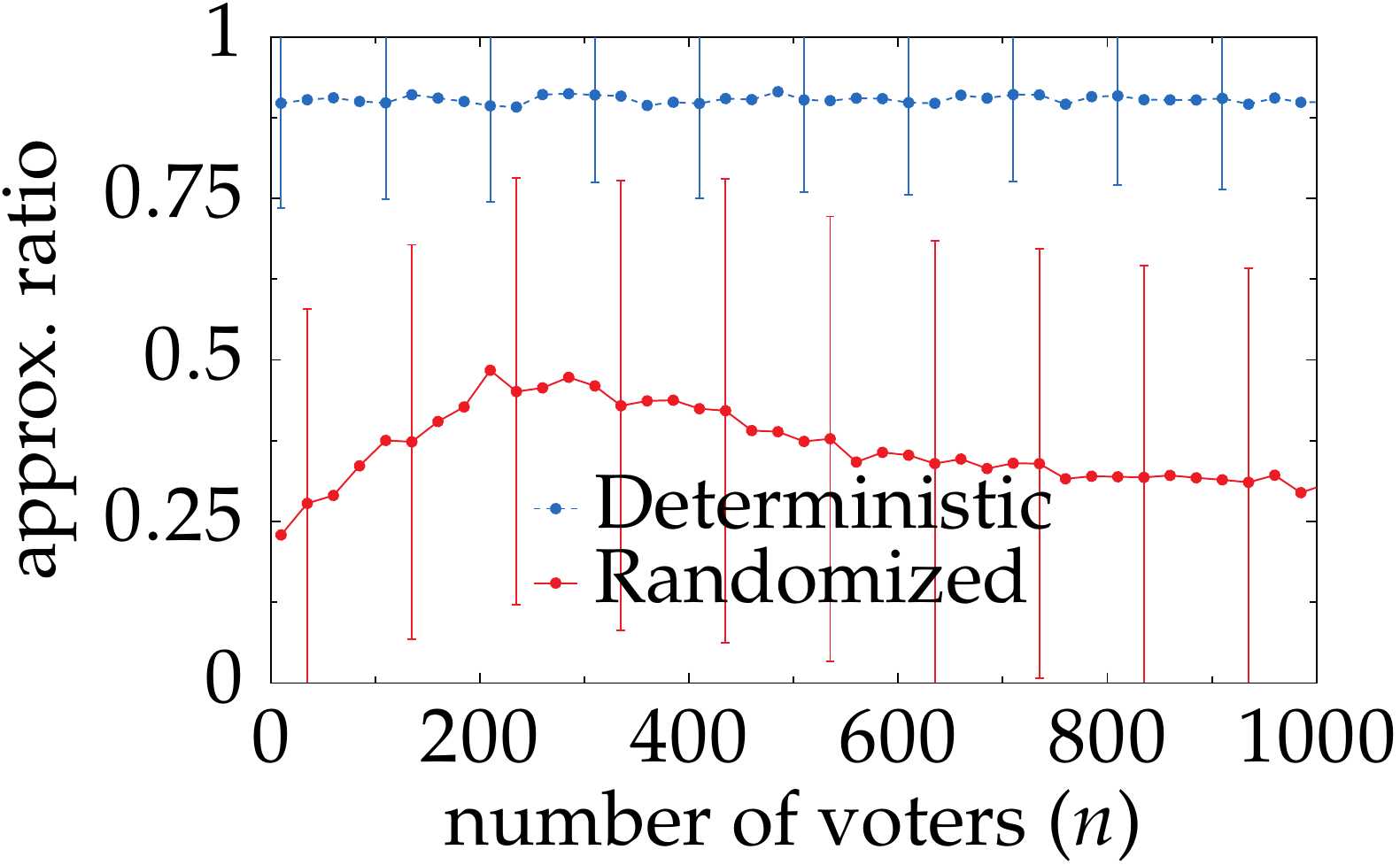}
  (a) $\ell = 2$
\endminipage\hfill
\minipage{0.42\textwidth}
  \includegraphics[width=\linewidth]{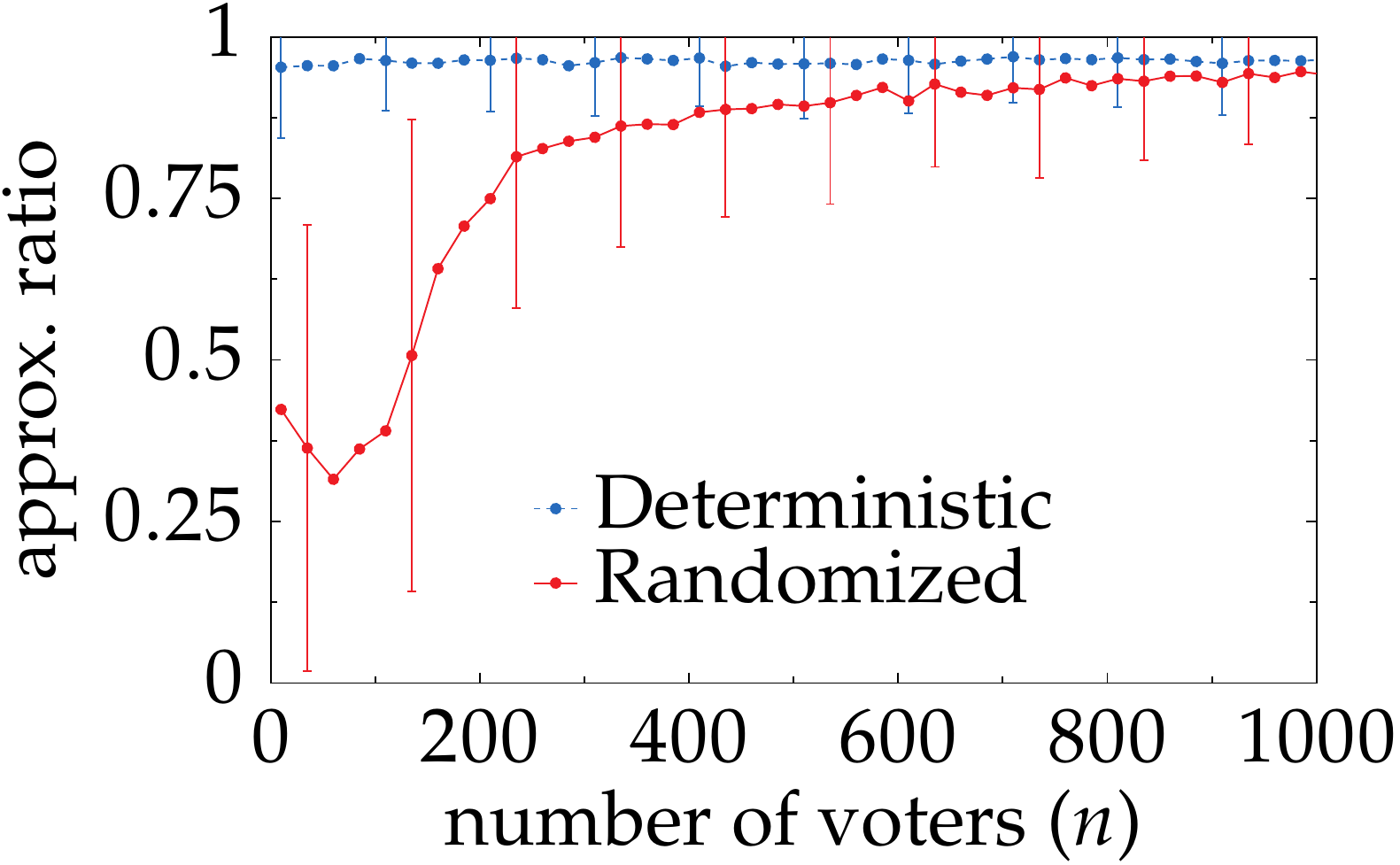}
  (b) $\ell = 8$
\endminipage

Single Peaked Impartial Culture

\minipage{0.42\textwidth}
  \includegraphics[width=\linewidth]{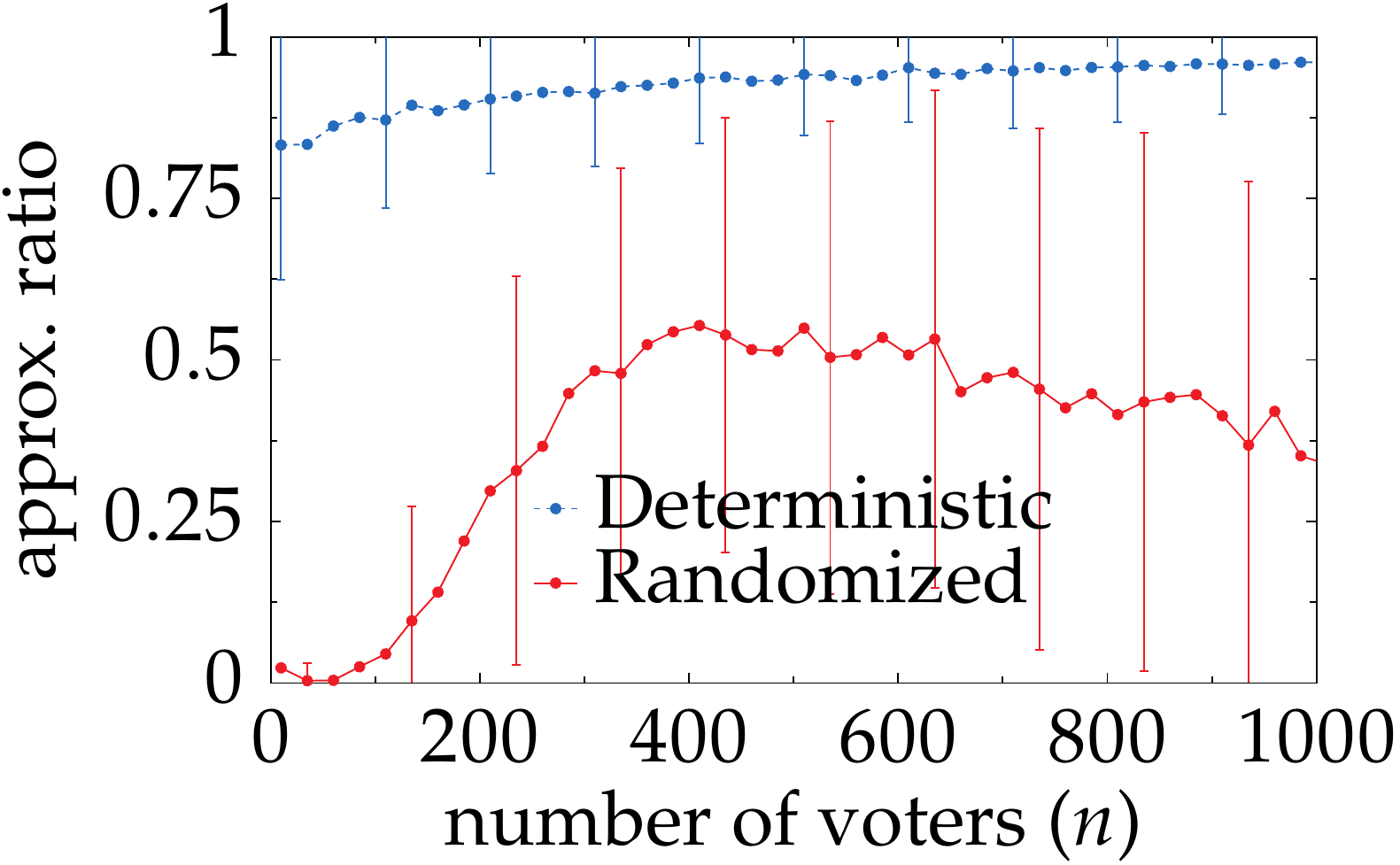}
  (a) $\ell = 2$
\endminipage\hfill
\minipage{0.42\textwidth}
  \includegraphics[width=\linewidth]{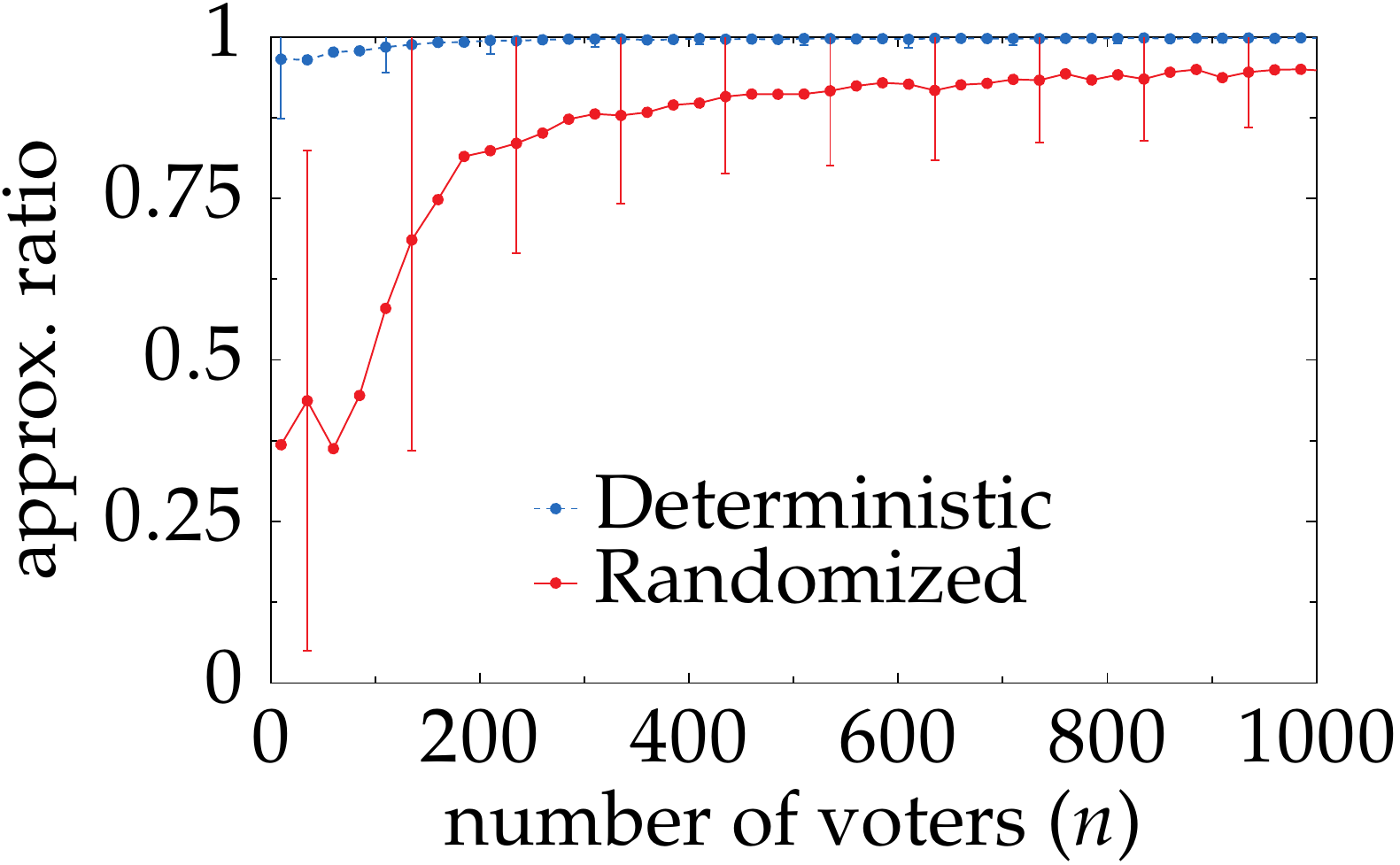}
  (b) $\ell = 8$
\endminipage
\end{center}
\caption{Approximation ratio for the two algorithms for the Minimax rule (the randomized algorithm asks to compare $\ell$ random candidates, and the deterministic asks for the $\ell$-truncated ballot) assessed through computer simulations.}\label{fig:approx_exper_minimax}
\end{figure}

\section{Conclusion}
In this paper we theoretically and experimentally analyzed how well certain election rules can be approximated when we are given only parts of voters' preferences. 
We compared two methods of eliciting voters' preferences, (i) the randomized method, where each voter is asked to compare a randomly selected subset of $\ell$ alternatives, and (2) the deterministic method, where we ask each voter to provide a ranking of her $\ell$ most preferred candidates. We investigated how well one can approximate positional scoring rules and the Minimax method through one of these two elicitation methods, providing both upper-bounds on the approximation ratio (impossibility results), and providing algorithms matching these bounds. 

We conclude that the randomized approach is usually superior; the exceptions include preference distributions with little or no structure, which rarely appear in practice. For the Borda rule, with hundreds of voters it is usually sufficient to ask each voter to compare two random candidates to achieve a high approximation guarantee. Approximating the Minimax rule is harder: one typically needs more voters and to ask them to compare more candidates---e.g., for $m=50$ candidates, we obtained high approximation guarantees for the Minimax rule only when we set the number of voters to around thousand and $\ell = 8$.
\subsubsection*{Acknowledgments}
Piotr Skowron was supported by a postdoctoral fellowship of the Alexander von Humboldt Foundation, Germany, and by the Foundation for Polish Science within the Homing programme (Project title: "Normative Comparison of Multiwinner Election Rules").

\bibliography{grypiotr2006}
\end{document}